\documentclass[3p,12pt,sort&compress,number]{elsarticle}
\biboptions{square,comma}
\usepackage[english]{babel}
\usepackage{graphicx}
\usepackage{hyperref}

\usepackage{enumerate}
\usepackage{amssymb}
\usepackage{amsmath}
\usepackage{amsthm}
\usepackage{todonotes}
\usepackage[labelfont=bf, compatibility=false, font=small]{caption}
\usepackage[labelfont=default]{subcaption}

\newcommand{\crn}{\ensuremath{\operatorname{cr}}}

\begin{document}
\begin{frontmatter}

\title{Compact Drawings of 1-Planar Graphs \\ with Right-Angle
  Crossings and Few Bends\tnoteref{t1}}

\tnotetext[t1]{An earlier version of this work appeared under the
same title in Proc.\ 26th Int.\ Sympos.\ Graph Drawing \& Network
Vis.\ (GD'18), pages 137--151, volume 11282 of Lect. Notes
Comp. Sci., Springer-Verlag~\cite{clwz-cd1pg-GD18}.}

\author{Steven Chaplick\fnref{orcidsc}}
\author{Fabian Lipp\fnref{orcidfl}}
\author{Alexander Wolff\fnref{orcidaw}}
\author{Johannes Zink\corref{cor1}\fnref{orcidjz}}

\address{Lehrstuhl f\"ur Informatik I, Universit\"at
      W\"urzburg, Germany \\
      \url{http://www1.informatik.uni-wuerzburg.de/en/staff} \\
      Manuscript submitted June 21, 2018;
      accepted February 11, 2019}

    \cortext[cor1]{Corresponding author; email:
      \href{mailto:zink@informatik.uni-wuerzburg.de}{\tt 
        zink@informatik.uni-wuerzburg.de}}
    \fntext[orcidsc]{ORCID:
      \href{https://orcid.org/0000-0003-3501-4608}{0000-0003-3501-4608}}
    \fntext[orcidfl]{ORCID:
      \href{https://orcid.org/0000-0001-7833-0454}{0000-0001-7833-0454}}
    \fntext[orcidaw]{ORCID:
      \href{https://orcid.org/0000-0001-5872-718X}{0000-0001-5872-718X}}
    	\fntext[orcidjz]{ORCID:
  	  \href{https://orcid.org/0000-0002-7398-718X}{0000-0002-7398-718X}}

\begin{keyword}
      graph drawing \sep beyond-planar graphs \sep 1-planar graphs
      \sep right-angle crossings \sep grid drawing
\end{keyword}

\newtheorem{theorem}{Theorem}
\newtheorem{lemma}[theorem]{Lemma}

\setcounter{tocdepth}{3}

\begin{abstract}
We study the following classes of beyond-planar graphs: 1-planar,
IC-planar, and NIC-planar graphs.  These are the graphs that admit a
1-planar, IC-planar, and NIC-planar drawing, respectively.
A drawing of a graph is \emph{1-planar} if every edge is crossed at most once.
A 1-planar drawing is \emph{IC-planar} if no two pairs of crossing edges share a vertex.
A 1-planar drawing is \emph{NIC-planar} if no two pairs of crossing edges share two vertices.

We study the relations of these beyond-planar graph classes (\emph{beyond-planar graphs} is a collective term for the primary attempts to generalize the planar graphs) to
\emph{right-angle crossing} (\emph{RAC}) graphs that
admit compact drawings on the grid with few bends.
We present four drawing algorithms that preserve the given embeddings.
First, we show that every $n$-vertex NIC-planar graph admits a NIC-planar RAC drawing with at most one bend per edge on a grid of size $\mathcal{O}(n) \times \mathcal{O}(n)$.
Then, we show that every $n$-vertex 1-planar graph admits a 1-planar RAC drawing with at most two bends per edge on a grid of size $\mathcal{O}(n^3) \times \mathcal{O}(n^3)$.
Finally, we make two known algorithms embedding-preserving; for
drawing 1-planar RAC graphs with at most one bend per edge and for
drawing IC-planar RAC graphs straight-line.
\end{abstract}

\end{frontmatter}

\graphicspath{{figures/}}

\section{Introduction}
\label{sec:introduction}

We frequently encounter \emph{graphs} and networks in diverse domains.
Taken by itself, a graph is just an abstract set of vertices with an abstract set of edges, which model some pairwise relations between vertices.
When it comes to visualizing graphs, one usually prefers aesthetically pleasing drawings of graphs, where the structure can be grasped quickly and unambiguously.
Several cognitive studies suggest that human viewers prefer drawings with few crossings~\cite{Purchase2000,Purchase2002,Ware2002}, no or few bends per edge~\cite{Purchase1997}, and large crossing angles if edges cross~\cite{Huang2007,Huang2014,Huang2008}.
Moreover a drawing should have a good resolution so that its elements are well recognizable and distinguishable.
A possibility to achieve this is placing every vertex, bend point, and crossing point onto a grid point of a small regular grid~\cite{Hoffmann2014}.

We focus on the aforementioned criteria.
In terms of large crossing angles, we only allow the largest possible crossing angle, which is 90 degrees.
In terms of edge crossings we restrict to some classes of \emph{beyond-planar graphs}, which is a collective term for the primary attempts to generalize or extend the well-known and deeply studied planar graphs.
In terms of grid size, it is natural to restrict to sizes
    that are polynomial in the input size.  This ensures that
    drawings of large instances are readable.
In particular, polynomials of low degree are desirable in this context.
Under these restrictions, we aim to minimize the number of bend points---in our work the maximum number of bends per edge.

Among the beyond-planar graphs, we consider the prominent class of \emph{1-planar graphs}, that is, graphs
that admit a drawing where each edge is crossed at most once, and some subclasses.
The 1-planar graphs were introduced by Ringel~\cite{Ringel1965} in 1965; Kobourov et
al.~\cite{KobourovLM17} surveyed them recently.
Regarding the right-angle crossings and the number of bend points, we will define \emph{$\text{RAC}_k$ graphs},
that is, graphs that admit a poly-line drawing where all
crossings are at right angles and each edge has at most $k$
bends. The $\text{RAC}_k$ graphs were introduced by
Didimo et al.~\cite{Didimo2011}.

We investigate the relationships between
(certain subclasses of) 1-planar graphs and $\text{RAC}_k$ graphs that admit
drawings on a polynomial-size grid.
The prior work and our contributions are summarized in
Fig.~\ref{fig:DiagramBendRACClassRelation}.
A broader overview of beyond-planar graph classes is given in a recent survey by Didimo et al.~\cite{Didimo2018}.

\paragraph{Basic Terminology}
\begin{figure}[t]
	\begin{subfigure}[t]{0.11\linewidth}
		\centering
		\includegraphics[page=2]
		{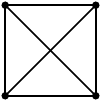}
		\caption{RAC$_0$ drawing}
		\label{fig:ExampleRACDrawing}
	\end{subfigure}
	\hfill
	\begin{subfigure}[t]{0.135\linewidth}
		\centering
		\includegraphics[page=3]
		{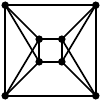}
		\caption{\mbox{IC-planar} drawing}
		\label{fig:ExampleIC-PlanarDrawing}
	\end{subfigure}
	\hfill
	\begin{subfigure}[t]{0.15\linewidth}
		\centering
		\includegraphics[page=4]
		{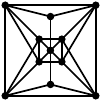}
		\caption{\mbox{NIC-planar} drawing}
		\label{fig:ExampleNIC-PlanarDrawing}
	\end{subfigure}
	\hfill
	\begin{subfigure}[t]{0.125\linewidth}
		\centering
		\includegraphics[page=5]
		{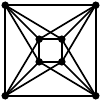}
		\caption{\mbox{1-planar} drawing}
		\label{fig:Example1-PlanarDrawing}
	\end{subfigure}
	\hfill
	\begin{subfigure}[t]{0.15\linewidth}
		\centering
		\includegraphics[page=6]
		{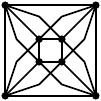}
		\caption{\mbox{1-planar} RAC$_1$ drawing}
		\label{fig:Example1-bend1-PlanarRACDrawing}
	\end{subfigure}
	\caption{Examples of different types of drawings.
		Figs.~\ref{fig:Example1-PlanarDrawing}
		and~\ref{fig:Example1-bend1-PlanarRACDrawing} show drawings
		of the same graph.
		Fig.~\ref{fig:Example1-bend1-PlanarRACDrawing} is taken from
		the Annotated Bibliography on
		1-Planarity~\cite{KobourovLM17}.}
	\label{fig:ExamplesOfDrawingTypes}
\end{figure}
A \emph{drawing} $\Gamma$ of a graph $G=(V,E)$ is a mapping that takes
each vertex $v \in V$ to a point in $\mathbb{R}^2$ and each edge $uv
\in E$ to a simple open Jordan curve in $\mathbb{R}^2$ such that the
endpoints of this curve are $\Gamma(u)$ and $\Gamma(v)$.
For convenience, we will refer to the images of vertices and edges
under~$\Gamma$ as vertices and edges, too.
The topologically connected regions of~$\mathbb{R}^2 \setminus \Gamma$ are the \emph{faces} of~$\Gamma$.
The infinite face of~$\Gamma$ is its \emph{outer} face;
the other faces are \emph{inner} faces.
Each face defines a circular list of bounding edges (resp. edge sides), which we call its \emph{boundary list}.
Two drawings of a graph~$G$ are \emph{equivalent} when they have the same set of boundary lists for their inner faces and the boundary list of the outer face of the one drawing equals the boundary list of the outer face of the other drawing.
Each equivalence class of drawings of $G$ is an \emph{embedding}.
We say that a drawing~$\Gamma$ \emph{respects} an embedding~$\mathcal{E}$ if~$\Gamma$ belongs to the equivalence class~$\mathcal{E}$ of drawings.
A \emph{$k$-bend (poly-line) drawing} is a drawing in which every edge
is drawn as a connected sequence of at most $k+1$ line segments.
The (up to)~$k$ inner vertices of an edge connecting these line segments are called \emph{bend points} or \emph{bends}.
A 0-bend drawing is more commonly referred to as a \emph{straight-line} drawing.
A \emph{drawing on the grid of size $w \times h$} is a drawing where
every vertex, bend point, and crossing point has integer coordinates
in the range $[0, w] \times [0, h]$.
In any drawing, we require that vertices, bends, and crossings are pairwise distinct points.

A drawing is \emph{1-planar} if every edge is crossed at most once.
A 1-planar drawing is \emph{independent-crossing planar}
(\emph{IC-planar}) if no two pairs of crossing edges share a vertex.
A 1-planar drawing is \emph{near-independent-crossing planar}
(\emph{NIC-planar}) if any two pairs of crossing edges share at most one
vertex.
A drawing is \emph{right-angle-crossing} (\emph{RAC}) if $(i)$ it is a poly-line drawing, $(ii)$~no more than two edges cross in the same point, and $(iii)$ in every crossing point the edges intersect at right angles.
We further specialize the notion of RAC drawings.
A drawing is RAC$_k$ if it is RAC and $k$-bend; it is
RAC$^\text{poly}$ if it is RAC and on a grid 
whose size is polynomial in its number of vertices.
Examples of IC-planar, NIC-planar, 1-planar, and RAC drawings are given in Fig.~\ref{fig:ExamplesOfDrawingTypes}.
The \emph{planar}, \emph{1-planar}, \emph{NIC-planar}, \emph{IC-planar}, and \emph{RAC$_k$} graphs are the graphs that admit a crossing-free, 1-planar, NIC-planar, IC-planar, and RAC$_k$ drawing, respectively.
More specifically, RAC$_k^\text{poly}$ is the set of graphs that admit a RAC$_k^\text{poly}$ drawing.
A \emph{plane}, \emph{1-plane}, \emph{NIC-plane}, and \emph{IC-plane}
graph is a graph given with a specific planar, 1-planar, NIC-planar,
and IC-planar embedding, respectively.
In a 1-planar embedding, the edge crossings are known and they are stored as if they were vertices.

We will denote an embedded graph by $(G, \mathcal{E})$ where $G$ is the graph and $\mathcal{E}$ is the embedding of this graph.
For a point~$p$ in the plane, let~$x(p)$ and~$y(p)$ denote its
x- and y-coordinate, respectively.  Given two points~$p$ and~$q$, we
denote the straight-line segment connecting them by~$\overline{pq}$
and its length, the Euclidean distance of~$p$ and~$q$,
by~$\|\overline{pq}\|$.

\begin{figure}[t]
	\centering
	\includegraphics[page=3]{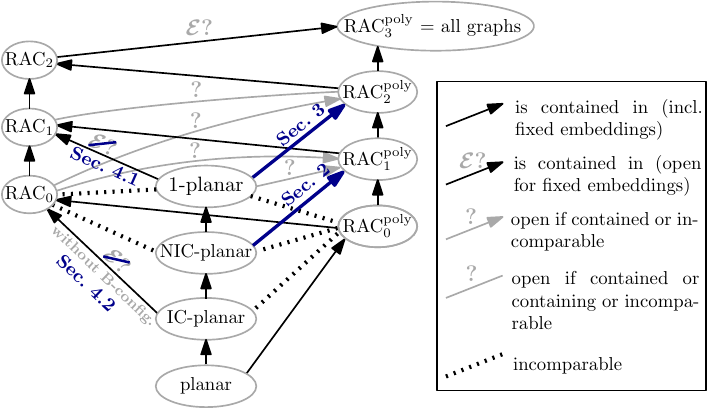}
	\caption{Relating some classes of (beyond-)planar graphs and RAC graphs. Our main results are the containment relationships indicated by the thick blue arrows.}
	\label{fig:DiagramBendRACClassRelation}
\end{figure}

\paragraph{Previous Work}

In the diagram in Fig.~\ref{fig:DiagramBendRACClassRelation}, we give
an overview of the relationships between classes of 1-planar graphs
and $\text{RAC}_k$ graphs.
Clearly, the planar graphs are a subset of the IC-planar graphs, which are a subset of the NIC-planar graphs, which are a subset of the 1-planar graphs.
It is well known that every plane graph can be drawn with straight-line edges on a grid of quadratic size~\cite{s-epgg-SODA90,Fraysseix1990}.
Every IC-planar graph admits an IC-planar $\text{RAC}_0$
drawing but not necessarily in polynomial area~\cite{Brandenburg2016}.
Moreover, there are graphs in $\text{RAC}_0^\text{poly}$ that
are not 1-planar~\cite{Eades2013} and, therefore, also not IC-planar.
The class of $\text{RAC}_0$
graphs is incomparable with the classes of NIC-planar
graphs~\cite{Bachmaier2017} and 1-planar graphs~\cite{Eades2013}.
Bekos et al.~\cite{Bekos2017} showed that every 1-planar graph admits a 1-planar RAC$_1$ drawing, but their recursive drawings may need exponential area.
Every graph admits a $\text{RAC}_3$ drawing in polynomial area, but this does not hold if a given embedding of the
graph must be preserved~\cite{Didimo2011}.

Later, our main tool will be the variant of the \emph{shift algorithm} by Harel and Sardas~\cite{Harel1998},
which is a generalization of the algorithm of Chrobak and Payne~\cite{Chrobak1995}, which in turn is based on the classical shift algorithm of de Fraysseix et al.~\cite{Fraysseix1990}.
The algorithm of Harel and Sardas runs in linear time and consists of two phases.
Given an $n$-vertex biconnected plane graph~$G$, in the first phase a \emph{biconnected canonical ordering}~$\Pi$ is computed.
\emph{Biconnected canonical orderings} are a generalization of canonical orderings (as used in the classical shift algorithm) that assume only biconnectivity instead of triconnectivity.
A biconnected canonical ordering~$\Pi = (v_1, v_2, \dots, v_n)$ is an
ordering of the vertices of~$G$ such that, for every $k \in \{2,3,\dots,n\}$,
\begin{itemize}
	\item the subgraph~$G_k$ of~$G$ induced by $\Pi_k := (v_1, \dots, v_{k})$ is connected,
	\item the edge~$v_1v_2$ lies on the boundary of the outer face of~$G_k$, and
	\item all vertices in $G - G_k$ lie within the outer face of~$G_k$.
\end{itemize}
Moreover, for every $k \ge 3$,
\begin{itemize}
\item the vertex~$v_k$ has one or more neighbors in~$G_{k-1}$.
  If~$v_k$ has exactly one neighbor~$u$ in~$G_{k-1}$, then~$v_k$ has a
  \emph{legal support}, that is, $u$ has a neighbor~$w$ on the
  boundary of the outer face of~$G_{k-1}$ that immediately follows or
  precedes~$v_k$ in the circular order of the neighbors of~$u$.
\end{itemize}
We call $w$ a \emph{support vertex} and the edge $uw$ a \emph{support edge} of~$v_k$.
Moreover, we say that a vertex~$u$ is \emph{covered} by~$v_k$ if~$u$ is on the boundary of the outer face of~$G_{k-1}$, but not on that of~$G_k$.
For a vertex~$v$, we define the set $C(v)$ as the set of all vertices
$u_1, \dots, u_m$ that are covered by~$v$, and we define 
$L(v)=C(v) \cup L(u_1) \cup \dots \cup L(u_m)$.
The set~$L(v)$ can be seen as the set of all vertices ``below''~$v$.
Unlike the classical shift algorithm, the algorithm of Harel and Sardas computes the (biconnected) canonical ordering bottom-up, which we will exploit in Section~\ref{sec:NIC-planar1-bendRACQuadraticArea}.
In the second phase, $G$ is drawn according to~$\Pi$.
This is done by an incremental drawing procedure using shifts on an integer grid known from the classical shift algorithm.
It starts by placing the vertices $v_1$, $v_2$, and $v_3$ onto the grid points $(0, 0)$, $(2, 0)$, and $(1, 1)$, respectively.
For every $k \in \{4,5,\dots,n\}$, the vertex~$v_k$ has a sequence of
at least two adjacent vertices or support vertices $u_1, \dots, u_m$
on the boundary of~$G_{k-1}$.
To add~$v_k$ to the drawing, all drawn vertices except for~$u_1$,
$L(u_1)$ and the vertices to the left of~$u_1$ (and the vertices below
them) are shifted to the right by one unit.
Then, $u_m$, $L(u_m)$ and all vertices to the right of~$u_m$ (and the
vertices below them) are again shifted to the right by one unit.
These shifts do not cause any new crossings.
After the shifts, $v_k$ is placed on the intersection of the line with
slope~$1$ through~$u_1$ and the line with slope~$-1$ through~$u_m$.
This intersection point is always a free grid point from which one can
``see'' $u_1, \dots, u_m$.
Therefore, the resulting drawing is a crossing-free straight-line drawing on a grid of size \mbox{$(2 |V(G)| - 4) \times (|V(G)| - 2)$}.

\paragraph{Our Contributions}
We contribute four new results;
two main results and two adaptations of prior results.

First, we constructively show that every NIC-plane graph admits a RAC$_1$ drawing in quadratic area; see Section~\ref{sec:NIC-planar1-bendRACQuadraticArea}.
This improves upon a side result by Liotta and Montecchiani~\cite{Liotta2016}, who showed that every IC-plane graph admits a RAC$_2$ drawing on a grid of quadratic size.
In Fig.~\ref{fig:full_drawing-after} we give a full example of a NIC-plane graph drawn with this algorithm.

Second, we constructively show that every 1-plane graph admits a RAC$_2$ drawing in polynomial area; see Section~\ref{sec:1-planar2-bendRACQuadraticArea}.

Beside these two main results, we show how to preserve a given
embedding when computing RAC drawings; more precisely
\begin{itemize}
	\item when computing RAC$_1$ drawings of 1-plane graphs (by
	adapting an algorithm of Bekos et al.~\cite{Bekos2017}; see
	Section~\ref{sec:1-planar1-bendRACSameEmbedding})
	and
	\item when computing RAC$_0$ drawings of IC-plane graphs (by adapting
	an algorithm of Brandenburg et al.~\cite{Brandenburg2016}; see
	Section~\ref{sec:IC-planar0-bendRACSameEmbedding}).
	Note that Thomassen~\cite{Thomassen1988} has characterized the
	straight-line drawable 1-plane graphs by certain forbidden
	configurations.  Therefore, this adaptation works only for IC-plane
	graphs without these configurations, i.e., for all straight-line
	drawable IC-plane graphs. 
\end{itemize}
\newcounter{theoremCounterOnePlaneOneBendRAC}
\setcounter{theoremCounterOnePlaneOneBendRAC}{\thetheorem} 
\newcommand{\theoremOnePlaneOneBendRAC}{%
	Any $n$-vertex 1-plane graph admits an embedding-preserving
	RAC$_1$ drawing. It can be computed in~$\mathcal{O}(n)$~time.
}%
\newcounter{theoremCounterICPlaneStraightLineRAC}
\setcounter{theoremCounterICPlaneStraightLineRAC}{\thetheorem} 
\newcommand{\theoremICPlaneStraightLineRAC}{%
	Any straight-line drawable $n$-vertex IC-plane graph admits an em\-bed\-ding-pre\-ser\-ving RAC$_0$ drawing.
	It can be computed in~$\mathcal{O}(n^3)$~time.
}

\section{NIC-Planar 1-Bend RAC Drawings in Quadratic Area}
\label{sec:NIC-planar1-bendRACQuadraticArea}

In this section we constructively show that quadratic area is
sufficient for RAC$_1$ drawings of NIC-planar graphs.  We prove the
following.
\begin{figure}[tb]
	\begin{subfigure}[t]{0.2 \textwidth}
		\centering
		\includegraphics[page=1]{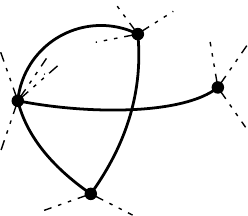}
		\caption{crossing as it initially appears}
		\label{fig:UncoveredCrossing}
	\end{subfigure}
	\hfill
	\begin{subfigure}[t]{0.27 \textwidth}
		\centering
		\includegraphics[page=2]{preprocessing_modification}
		\caption{empty kite and subdivided original edge}
		\label{fig:EmptyKite}
	\end{subfigure}
	\hfill
	\begin{subfigure}[t]{0.21 \textwidth}
		\centering
		\includegraphics[page=3]{preprocessing_modification}
		\caption{empty quadrangle}
		\label{fig:EmptyQuadrangle}
	\end{subfigure}
	\hfill
	\begin{subfigure}[t]{0.21 \textwidth}
		\centering
		\includegraphics[page=4]{preprocessing_modification}
		\caption{divided quadrangle}
		\label{fig:DividedQuadrangle}
	\end{subfigure}
	
	\caption{Modifying the crossings and computing the biconnected canonical ordering.}
	\label{fig:Preprocessing}
\end{figure}%
\begin{theorem}
	\label{thm:NIC-planar1-bendRACQuadraticArea}
	Any $n$-vertex NIC-plane graph $(G, \mathcal{E})$ admits a
	NIC-planar RAC$_1$ drawing that respects~$\mathcal{E}$ and lies
	on a grid of size $\mathcal{O}(n) \times \mathcal{O}(n)$.
	The drawing can be computed in~$\mathcal{O}(n)$~time.
\end{theorem}%
In Fig.~\ref{fig:full_drawing-after} we show a complete drawing of a NIC-plane graph generated by our algorithm.
\paragraph{Preprocessing}
Our algorithm gets an $n$-vertex NIC-plane graph $(G, \mathcal{E})$ as input.
We first aim to make $(G, \mathcal{E})$
biconnected and planar so that we can draw it using the algorithm by
Harel and Sardas~\cite{Harel1998}.
Around each crossing in~$\mathcal{E}$, we insert up to four dummy
edges to obtain \emph{empty kites}.  A \emph{kite} is a $K_4$ that is
embedded such that (i)~every vertex lies on the boundary of the outer
face, and (ii)~there is exactly one crossing, which
does not lie on the boundary of the outer face.
A~kite~$K$ as a subgraph of an embedded graph $H$ is said to be \emph{empty} if there is no edge of~$H \backslash K$ that is on an inner face of~$K$ or crosses edges of $K$.
Inserting a dummy edge could create a pair of parallel
edges. If this happens, we subdivide the original edge participating in this
pair by a dummy vertex (see the transition from
Fig.~\ref{fig:UncoveredCrossing} to~\ref{fig:EmptyKite}). 
Note that we never create parallel dummy edges because $G$ is NIC-planar, which means that two different crossings have at most one vertex in common, but we would need two common vertices for a dummy edge to be duplicated.
After this, we remove both crossing edges from each empty kite and
obtain \emph{empty quadrangles} (see Fig.~\ref{fig:EmptyQuadrangle}).
We store each such empty quadrangle in a list~$Q$.
At the end of the preprocessing, we make the resulting plane graph
biconnected via, e.g., the algorithm of Hopcroft and
Tarjan~\cite{Hopcroft1973}.
Since each empty quadrangle is contained in a biconnected component, no edges are inserted into it.
Let~$(G', \mathcal{E}')$ be the resulting plane biconnected graph.

\paragraph{Drawing Step}
Now, we draw a graph that we obtain from $(G', \mathcal{E}')$ by using the variant of the shift algorithm by Harel and Sardas~\cite{Harel1998}.
We will exploit that it computes the biconnected canonical ordering bottom-up instead of top-down.

We perform the following additional operations when we compute the
biconnected canonical ordering~$\hat{\Pi}$.  Whenever we reach an empty quadrangle
$q = (a, b, c, d)$ of the list~$Q$ for
the first time, i.e., when the first vertex of $q$---say~$a$---is
added to the biconnected canonical ordering, we insert an
edge inside $q$ from $a$ to the vertex opposite~$a$ in~$q$, that is, to~$c$. 
We call the resulting structure a \emph{divided quadrangle} (see
Fig.~\ref{fig:DividedQuadrangle}).
In two special cases, we perform further modifications of the graph.
They will help us to guarantee a correct reinsertion of the
crossing edges in the next step of the algorithm.
Namely, when we encounter the last vertex~$v_\text{last} \in
\{b,c,d\}$ of~$q$, we distinguish three cases.
\begin{description}
	\item[Case~1:] $v_\text{last} = c$ (see Fig.~\ref{fig:Case1-before}).
		
	In this case, no extra operations are performed.
	
	\item[Case~2:] $v_\text{last} \in \{b,d\}$, and the other of $\{b,d\}$ is covered by~$c$ (see Fig.~\ref{fig:Case2-before}).
	
	We insert a dummy vertex~$v_\text{shift}$, which we call \emph{shift
		vertex}, into the current biconnected canonical ordering directly
	before~$v_\text{last}$ and make it adjacent to~$a$ and~$c$.
	Observe that, if $v_\text{shift}$ is the $k$-th vertex in~$\hat{\Pi}$, this still yields a valid biconnected canonical ordering since $v_\text{shift}$ has both neighbors in~$\hat{\Pi}_{k-1}$ and is on the outer face of the subgraph induced by~$\hat{\Pi}_{k-1}$.
	Later, we will remove~$v_\text{shift}$, but for now it causes the
	algorithm of Harel and Sardas to shift~$a$ and~$c$ away from each
	other by two units because~$v_\text{shift}$ is only adjacent to~$a$ and~$c$ and effects a regular shift between them before $v_\text{last}$ is added.
	
	\item[Case~3:] $v_\text{last} \in \{b,d\}$, and neither~$b$
	nor~$d$ is covered by~$c$ (see Fig.~\ref{fig:Case3-before}).
	
	Let $\{v_\text{lower}\} = \{b,d\} \setminus v_\text{last}$.
	We subdivide the edge~$av_\text{lower}$ via a dummy vertex~$v_\text{dummy}$.
	If $av_\text{lower}$ is an original edge of the input graph, this edge will be bent at $v_\text{dummy}$ in the final drawing.
	We insert $v_\text{dummy}$ into the current biconnected canonical ordering directly before~$v_\text{lower}$.
	To obtain a divided quadrangle again, we insert the dummy edge~$av_\text{lower}$, which we will remove before we reinsert the
	crossing edges.
	This will give us some extra space inside the triangle $(a, v_\text{dummy}, v_\text{lower})$ for a bend point.
	Inserting $v_\text{dummy}$ as $k$-th vertex into~$\hat{\Pi}$ keeps $\hat{\Pi}$ valid since~$v_\text{dummy}$ uses the edge incident to~$a$ that precedes the subdivided edge in the circular order around~$a$ (and is in the subgraph induced by~$\hat{\Pi}_{k-1}$) as a support edge.
	This edge would have been covered by~$v_\text{lower}$ otherwise.
	Then, $v_\text{lower}$ has at least two neighbors in~$\hat{\Pi}_k$, namely $a$ and~$v_\text{dummy}$, and is on the outer face of the subgraph induced by~$\hat{\Pi}_k$.
	Hence, the biconnected canonical ordering remains valid.
\end{description}
We draw the resulting plane biconnected $\hat{n}$-vertex graph
$(\hat{G}, \hat{\mathcal{E}})$ according to its biconnected canonical ordering~$\hat{\Pi}$ via the algorithm by Harel and Sardas and obtain
a crossing-free drawing~$\hat{\Gamma}$.
We do not modify the actual drawing phase.

\begin{figure}[t]
	\begin{subfigure}[t]{0.25 \textwidth}
		\centering
		\includegraphics[page=1]{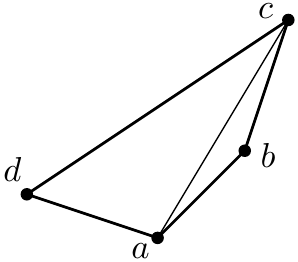}
		\caption{Case~1; $v_\text{last} = c$}
		\label{fig:Case1-before}
	\end{subfigure}
	\hfill
	\begin{subfigure}[t]{0.32 \textwidth}
		\centering
		\includegraphics[page=1]{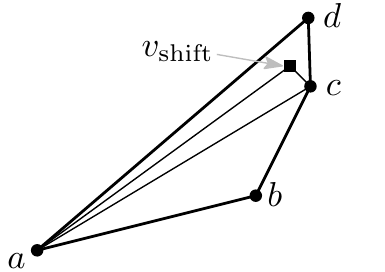}
		\caption{Case~2; $v_\text{last} = d$
			and $b$ is covered by~$c$}
		\label{fig:Case2-before}
	\end{subfigure}
	\hfill
	\begin{subfigure}[t]{0.34 \textwidth}
		\centering
		\includegraphics[page=1]{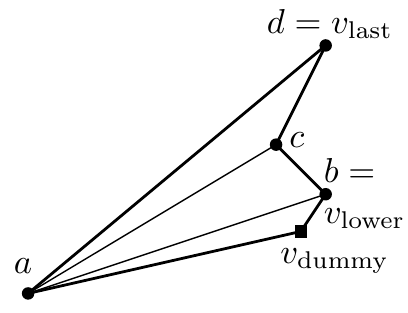}
		\caption{Case~3; $v_\text{last} = d$
			and $b$ is not covered by $c$}
		\label{fig:Case3-before}
	\end{subfigure}
	
	\bigskip
	
	\begin{subfigure}[t]{0.25 \textwidth}
		\centering
		\includegraphics[page=2]{nic_case_1}
		\caption{Case~1}
		\label{fig:Case1-after}
	\end{subfigure}
	\hfill
	\begin{subfigure}[t]{0.32 \textwidth}
		\centering
		\includegraphics[page=2]{nic_case_2}
		\caption{Case~2}
		\label{fig:Case2-after}
	\end{subfigure}
	\hfill
	\begin{subfigure}[t]{0.34 \textwidth}
		\centering
		\includegraphics[page=2]{nic_case_3}
		\caption{Case~3}
		\label{fig:Case3-after}
	\end{subfigure}
	
	\caption{Divided quadrangles produced in the three cases of the drawing step (a)--(c) and the crossing edges after the reinsertion step (d)--(f) in our algorithm.  For orientation, lines with slope $1$ or $-1$ are dashed violet.}
	\label{fig:CasesForEdgeInsertion}
\end{figure}

\paragraph{Postprocessing (Reinserting the Crossing Edges)}
We refine the underlying grid of~$\hat{\Gamma}$ by a factor of~$2$ in
both dimensions. 
Let $q = (a, b, c, d)$ be a quadrangle in~$Q$, where $a$ is the first
and $v_\text{last}$ the last vertex in $\hat{\Pi}$ among the vertices in~$q$. 
From~$q$, we first remove the chord edge~$ac$ and obtain an
empty quadrangle.  Then, we distinguish three cases for reinserting
the crossing edges that we removed in the preprocessing.  These are
the same cases as in the description of the modified computation of
the biconnected canonical ordering before.
In this case distinction we omit some lengthy but straight-forward
calculations; see Zink's master's
thesis~\cite{Zink2017} for the details.

\begin{description}
	\item[Case~1:] $v_\text{last} = c$ (see Fig.~\ref{fig:Case1-before}).
	
	Since $c$ is adjacent to $a$, $b$, and $d$ in $\hat{G}$ and was added to the drawing after them, it has the largest $y$-coordinate among the vertices in $q$.
	Assume that $y(d)$ is smaller or equal to~$y(b)$ since the other case is symmetric.
	An example of a quadrangle in this case before and after the reinsertion of the crossing edges is given in Figs.~\ref{fig:Case1-before} and~\ref{fig:Case1-after}, respectively.
	We will have a crossing point at $(x(a), y(d))$.
	To this end, we insert the edge~$ac$ with a bend at $e_{ac} = (x(a), y(d) + 1)$ and we insert the edge~$bd$ with a bend at $e_{bd} = (x(a) + 1, y(d))$.
	Clearly, the crossing is at a right angle.
	Observe that $q$ is convex since $c$ is the last drawn vertex of~$q$, $c$ is adjacent to each of $(b, a, d)$ in this circular order in the embedding, $a$ was drawn first, and $a$ is also adjacent to~$b$ and~$d$.
	Moreover, observe that both bend points lie inside~$q$.
	Therefore, it follows that both crossing edges lie completely inside~$q$.
	
	\item[Case~2:] $v_\text{last} \in \{b,d\}$, and the other of $\{b,d\}$ is covered by~$c$ (see Fig.~\ref{fig:Case2-before}).
	
	Assume that $y(d)>y(b)$; the other case is symmetric.
	An example of a quadrangle in this case before and after the
	reinsertion of the crossing edges is given in
	Figs.~\ref{fig:Case2-before} and~\ref{fig:Case2-after}, respectively.
	We remove $v_\text{shift}$ in addition to removing the edge $ac$.
	We define the crossing point~$p_\text{cross} = (x_\text{cross}, y_\text{cross})$ as the intersection point of the lines with slope $1$ and $-1$ through $c$ and $b$, respectively.
	The coordinates of this crossing point are $x_\text{cross} = (x(c) -
	y(c) + x(b) + y(b))/2$ and $y_\text{cross} = (-x(c) + y(c) + x(b) +
	y(b))/2$.
	Since we refined the grid by a factor of~2 in each dimension, the above coordinates are both integers.
	We place the two bend points onto the same lines at the closest grid points that are next to~$p_\text{cross}$,~i.e., 
	we draw the edge~$ac$ with a bend point
	at~$e_{ac} = (x_\text{cross} - 1, y_\text{cross} - 1)$ and we
	insert the edge $bd$ with a bend point at~$e_{bd} =
	(x_\text{cross} - 1, y_\text{cross} + 1)$.
	We do not intersect or touch the edge~$ad$ because we
	shifted~$a$ far enough away from~$c$ by the extra shift due
	to~$v_\text{shift}$.
	Moreover, the points~$e_{ac}$ and~$p_\text{cross}$ on the line with slope $1$ through $c$ are inside the empty quadrangle~$q$ since $b$ is covered by~$c$ (then~$b$ is below the line with slope~$1$ through~$c$) and $y(b)$ is at most equal to~$y(e_{ac})$.
	
	\item[Case~3:] $v_\text{last} \in \{b,d\}$, and neither~$b$	nor~$d$ is covered by~$c$ (see Fig.~\ref{fig:Case3-before}).
	
	Assume that $y(d)>y(b)$; again, the other case is symmetric.
	An example of a quadrangle in this case before and after the
	reinsertion of the crossing edges is given in
	Figs.~\ref{fig:Case3-before} and~\ref{fig:Case3-after}, respectively.
	Note that the edge~$ab$ is a dummy edge, which we inserted during the computation of~$\hat{\Pi}$, and next to this edge, there is the path~$av_\text{dummy}b$.
	This path is the former edge~$ab$.
	We will reinsert the edges $ac$ and $bd$ such that they cross in $(x(c), y(b))$.
	We will bend the edge $bd$ on the line with slope $1$ through $c$ at $y = y(b)$ because from this point we always ``see'' $d$ inside $q$.
	So, we define $x_\text{bend} := x(c) - \Delta y$ with $\Delta y := y(c) - y(b)$.
	First, we remove the dummy edge~$ab$.
	Second, we insert the edge~$ac$ with a bend point at~$e_{ac} = (x(c), y(b) - 1)$.  
	Third, we insert the edge~$bd$ with a bend point at~$e_{bd} = (x_\text{bend}, y(b))$.
	Note that $e_{ac}$ might be below the straight-line segment $\overline{a b}$ since $a$ could have been shifted far away from~$c$.
	However, $e_{ac}$ cannot be on or below the path~$av_\text{dummy}b$ because $y(v_\text{dummy})<y(e_{ac})$ and
	the slope of the line segment~$\overline{v_\text{dummy} b}$ is either
	greater than~$1$ or negative.
	Therefore, the crossing edges $ac$ and $bd$ lie completely
	inside the pentagonal face $(a, v_\text{dummy}, b, c, d)$.
\end{description}

\paragraph{Result}
After we have reinserted the crossing edges into each quadrangle
of~$Q$, we remove all dummy edges and transform the remaining dummy
vertices to bend points.
The resulting drawing~$\Gamma$ is a $\text{RAC}_1$ drawing that preserves
the embedding of the NIC-plane input graph~$(G,\mathcal{E})$.
We state the following lemmas regarding the graph size after inserting dummy vertices and the finally used grid size.
\begin{lemma}
	\label{lem:O(n)=O(n')=O(hat(n))}
	Let $G$ be the input graph, $G'$ be the graph after the preprocessing, and $\hat{G}$ be the graph after the computation of the biconnected canonical ordering.
	Let $n$, $n'$, and~$\hat{n}$ be the number of vertices in $G$, $G'$, and $\hat{G}$, respectively.
	It holds that $n' \leq 3.4 n - 4.8$ and $\hat{n} \leq 4 n - 6$.
\end{lemma}

\begin{proof}
	In the first step of the preprocessing, we create empty kites around every crossing.
	By creating the empty kites for every crossing, there are edges added to the graph (but we do not count edges here) and there are edges subdivided.
	When we subdivide an edge, we add a new vertex.
	There are at most four edges per crossing that are subdivided.
	The number~$\crn(\mathcal{E})$ of crossings in an $n$-vertex NIC-planar embedding $\mathcal{E}$ is bounded by $0.6 n - 1.2$~\cite{Czap2014,Zhang2014}.
	Using this, we can bound the number $n_\text{subdivide}$ of vertices that are added in this step to:
	\begin{equation*}
	n_\text{subdivide} \leq 4 \crn(\mathcal{E}) \leq 2.4 n - 4.8
	\end{equation*}
	In the second step of the preprocessing, we make the graph biconnected.
	To accomplish this, we only insert edges and the number of vertices does not increase. 
	So the number $n'$ of vertices of the graph $G'$ is:
	\begin{equation*}
	n' = n + n_\text{subdivide} \leq 3.4 n - 4.8
	\end{equation*}
	While computing the biconnected canonical ordering~$\hat{\Pi}$, at most one dummy vertex is added per crossing---either as shift vertex in Case~2 or as dummy vertex in Case~3.
	So the number $n_{\hat{\Pi}}$ of vertices added there is:
	\begin{equation*}
	n_{\hat{\Pi}} \leq \crn(\mathcal{E}) \leq 0.6 n - 1.2
	\end{equation*}
	And in total:
	\begin{equation*}
	\hat{n} = n' + n_{\hat{\Pi}} \leq (3.4 n - 4.8) + (0.6 n - 1.2) = 4 n - 6 \qedhere
	\end{equation*}
\end{proof}

Thus, only linearly many new vertices are added when constructing~$G'$ from~$G$ and~$\hat{G}$ from~$G'$.

\newcounter{lemmaCounterGridSizeNICPlanarRAC}
\setcounter{lemmaCounterGridSizeNICPlanarRAC}{\thetheorem} 
\newcommand{\lemmaGridSizeNICPlanarRAC}{%
	Every vertex, bend point, and crossing point of the drawing
	returned by our algorithm lies on a grid of size at most
	$(16n - 32) \times (8n - 16)$.
} 
\begin{lemma}
	\label{lem:NICPlane1BendAlgoBoundToQuadraticSize}
	\lemmaGridSizeNICPlanarRAC
\end{lemma}

\begin{proof}
	The shift algorithm places every vertex of the graph $\hat{G} = (\hat{V}, \hat{E})$ onto a grid point of a grid of size $(2 \hat{n} - 4) \times (\hat{n} - 2)$.
	By the upper bound on $\hat{n}$ from Lemma~\ref{lem:O(n)=O(n')=O(hat(n))}, we get the following grid size:
	\begin{align*}
	\text{coarser grid size} &\leq (2 (4 n - 6) - 4) \times ((4 n - 6) - 2) \nonumber \\
	&= (8 n - 16) \times (4 n - 8)
	\end{align*}
	This grid is later refined by a factor of $2$ in both dimensions.
	This bounds the size of the grid as follows:
	\begin{align*}
	\text{total grid size} &\leq (2 (8 n - 16)) \times (2(4 n - 8)) \nonumber \\
	&= (16 n - 32) \times (8 n - 16)
	\end{align*}
	We place bend points onto grid points on inner faces only.
	So, the total size of the drawing and its underlying grid does not increase when we add them.
\end{proof}
The shift algorithm of Harel and Sardas runs in linear
time~\cite{Harel1998}.  Also,
our additional operations can be performed in linear
time~(\cite{Zink2017}, Chapter~3.3.). 
This proves Theorem~\ref{thm:NIC-planar1-bendRACQuadraticArea}.
\paragraph{Full Example}
\begin{figure}
	\centering
	\begin{subfigure}[t]{\linewidth}
		\centering
		\includegraphics[page=1, width= \textwidth] {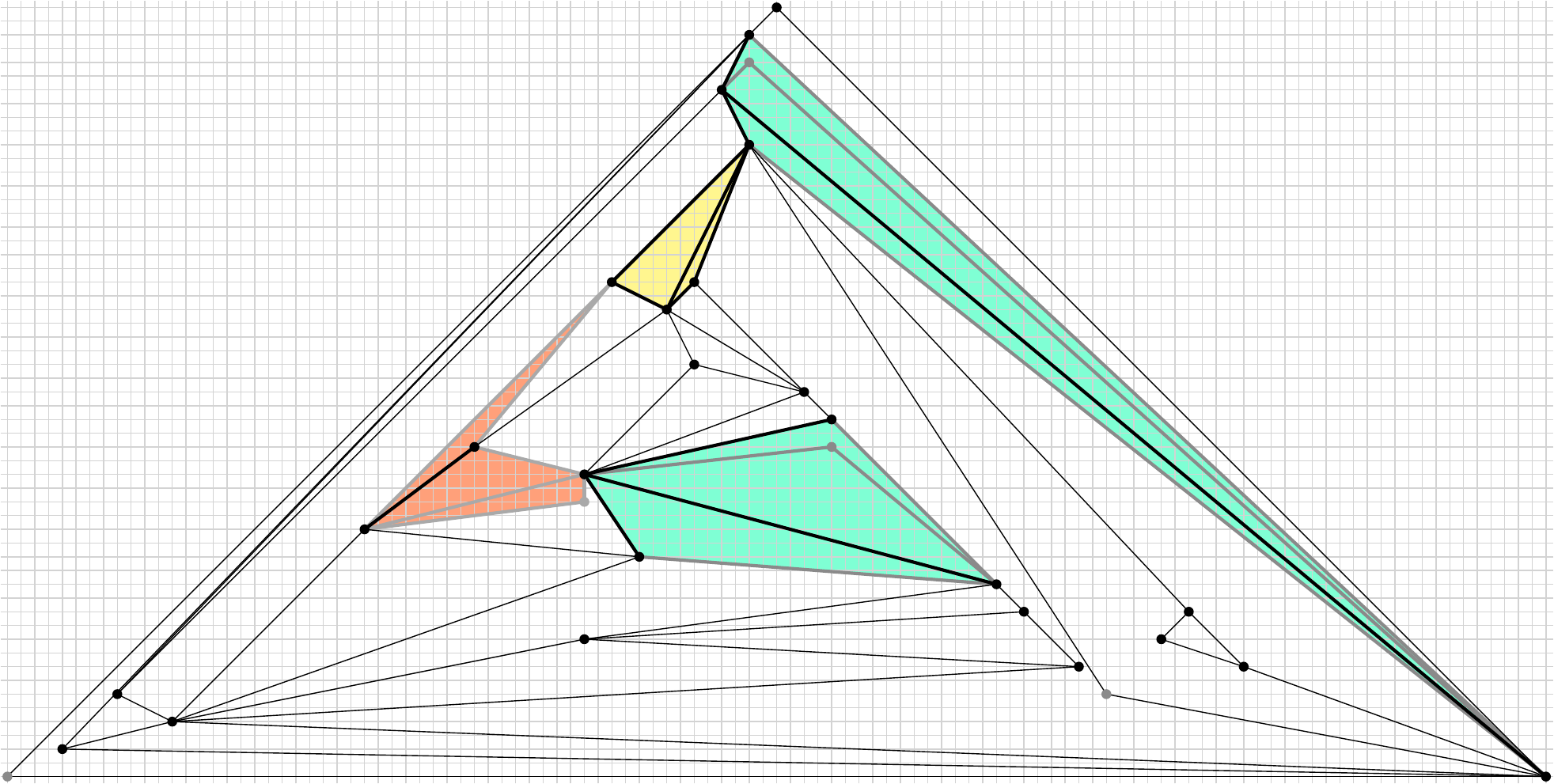}
		\caption{Intermediate drawing, generated by the shift algorithm, before the crossing edges are reinserted.}
		\label{fig:full_drawing-before}
	\end{subfigure}

	\bigskip

	\begin{subfigure}[t]{\linewidth}
		\centering
		\includegraphics[page=2, width= \textwidth] {full_drawing}
		\caption{Final NIC-planar RAC$_1$ drawing after the crossing edges have been reinserted and the dummy edges have been removed.}
		\label{fig:full_drawing-after}
	\end{subfigure}
	\caption{Full example computed by our Java implementation of the algorithm from Section~\ref{sec:NIC-planar1-bendRACQuadraticArea}.}%
	\label{fig:full_drawing}
\end{figure}
We have implemented our algorithm in Java.
The source code is available on GitHub~\cite{Zink2018}.
Figure~\ref{fig:full_drawing} shows a
drawing of a NIC-plane graph produced by this implementation.
The embedded graph in this example has four crossings.
For two of these crossings, Case~2 of our algorithm applies (green background color). 
Case~1 (yellow background color) and
Case~3 (red background color) apply to one crossing each.
In particular, in Fig.~\ref{fig:full_drawing-after}, two pairs of segments with slope $+1$ and $-1$ cross at a right angle and two pairs of horizontal/vertical segments cross.
The drawing in Fig.~\ref{fig:full_drawing-before} shows the graph as it
is drawn by the shift algorithm and before the crossing edges are
inserted.
The divided quadrangles into which we inserted crossing edges in the next step are highlighted by thick edges.
Dummy edges and vertices are drawn in gray.
Note that the two divided quadrangles in Case~2 contain an additional shift vertex and the one in Case~3 has an additional 2-path, which is also highlighted and makes the quadrangle a pentagon with a second chord edge.
The drawing in Fig.~\ref{fig:full_drawing-after} shows the final graph drawing after the crossing edges have been reinserted in the postprocessing step and after the dummy edges and vertices have been removed.
The four pairs of crossing edges are highlighted by thick edges.

\section{1-Planar 2-Bend RAC Drawings in Polynomial Area}
\label{sec:1-planar2-bendRACQuadraticArea}

In this section we constructively show that polynomial area is
sufficient for RAC$_2$ drawings of 1-planar graphs.  We prove the
following.

\begin{theorem}
	\label{thm:1-planar2-bendRACQuadraticArea}
	Any $n$-vertex 1-plane graph $(G, \mathcal{E})$ admits a 1-planar
	RAC$_2$ drawing that respects~$\mathcal{E}$
	and lies on a grid of size $\mathcal{O}(n^3) \times \mathcal{O}(n^3)$.
	The drawing can be computed in~$\mathcal{O}(n)$~time.
\end{theorem}

The idea of our algorithm is to draw a slightly modified, planarized version
of the 1-plane input graph with a variant of the shift algorithm
(by Harel and Sardas~\cite{Harel1998}) and then
``manually'' redraw the crossing edges so that they cross at right
angles and have at most two bends each.
An example of this redrawing of a crossing is depicted in Figs.~\ref{fig:inserting-bend-point-on-the-grid(c)}--\ref{fig:inserting-bend-point-on-the-grid(h)}.
The difficulty is to find grid points for the bend points and the crossings
so that the redrawn edges do not touch or cross the surrounding edges
drawn by the shift algorithm.
To this end, we refine our grid and place the middle part of each
crossing edge onto a horizontal or vertical grid line so that the edge
crossings are at right angles.

\begin{figure}[tb]
	\begin{subfigure}[t]{0.48 \linewidth}
		\centering
		\includegraphics[page=1] {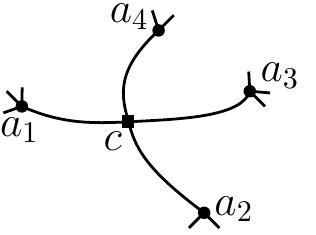}
		\caption{Planarized crossing where the crossing point became a crossing vertex $c$.}
		\label{fig:extendedEmptyKite1}
	\end{subfigure}
	\hfill
	\begin{subfigure}[t]{0.48 \linewidth}
		\centering
		\includegraphics[page=2] {extended_empty_kite}
		\caption{Enclosing the crossing vertex $c$ by a subdivided kite.}
		\label{fig:extendedEmptyKite2}
	\end{subfigure}
	
	\caption{%
		A crossing point is replaced by a crossing vertex $c$ and we insert four 2-paths of two dummy edges and a dummy vertex to induce a subdivided kite at each crossing.
		The vertices $d_1$, $d_2$, $d_3$, and $d_4$ are the dummy vertices of these 2-paths.
	}
	\label{fig:extendedEmptyKite}
\end{figure}

\paragraph{Preprocessing}
Our algorithm gets an $n$-vertex 1-plane graph $(G, \mathcal{E})$ as input.
First, we planarize $G$ by replacing each crossing point by a vertex (see Fig.~\ref{fig:extendedEmptyKite1}).
We will refer to the resulting vertices as \emph{crossing vertices}.
Second, we enclose each crossing vertex by a \emph{subdivided kite}, which is an empty kite where the four boundary edges are
subdivided by a vertex (see Fig.~\ref{fig:extendedEmptyKite2}).
We use subdivided kites instead of empty kites to maintain the embedding and to
avoid adding parallel edges.
Third, we make the graph biconnected using, e.g., the algorithm of Hopcroft and
Tarjan~\cite{Hopcroft1973}.
Note that we do not insert edges into inner faces of subdivided kites because all vertices and edges of a subdivided kite are in the same biconnected component.
After these three steps, we have a biconnected plane graph $(G', \mathcal{E}')$.
We draw $(G', \mathcal{E}')$ using the algorithm of Harel and
Sardas~\cite{Harel1998}.
This algorithm returns a crossing-free straight-line drawing~$\Gamma'$ of~$(G', \mathcal{E}')$, whose vertices lie on a grid of size $(2 n' - 4) \times (n'-2)$, where $n'$ is the number of vertices of~$G'$.

\paragraph{Assignment of Edges to Axis-Parallel Half-Lines}
For each crossing vertex~$c$, there are four incident edges in~$G'$.
They correspond to two edges of~$G$.
Consider the circular order around~$c$ in~$(G', \mathcal{E}')$.
The first and the third edge incident to $c$ correspond to one edge in~$(G,
\mathcal{E})$; symmetrically, the second and fourth incident edge correspond to
one edge.
To obtain a RAC drawing from this, we redraw each of the four edges around~$c$.
Consider an edge~$ac$ from a vertex~$a$ of the subdivided kite to the crossing vertex~$c$.
This edge is then redrawn with a bend point~$b$ that lies on an axis-parallel
line through~$c$.
In order to obtain a right-angle crossing, we bijectively assign the four incident edges
to the four axis-parallel half-lines originating in~$c$.
We call such a mapping an \emph{assignment}.
We do not take an arbitrary assignment, but take care to avoid
extra crossings with edges that are redrawn or previously drawn.
We call an assignment~$A$ \emph{valid} if there is a way to redraw each edge~$e$ with one bend so that the bend point of~$e$ lies on the half-line $A(e)$ and the resulting drawing is plane.

To ensure that our valid assignment can be realized on a small grid, we introduce further criteria.
We say that an edge~$e_1$ \emph{depends} on another edge~$e_2$ with
respect to an assignment~$A$ if $e_2$ lies in the angular sector between~$e_1$ and the
half-line $A(e_1)$.
In Fig.~\ref{fig:assignment_case_1}, for example, the edge~$e_3$
depends on~$e_4$ and $e_2$ depends on~$e_1$, but~$e_1$ and~$e_4$ do
not depend on any edge.  We call edges (such as~$e_1$ and~$e_4$) that
do not depend on other edges \emph{independent}.
We define the \emph{dependency depth} of an
assignment to be the largest integer~$k$ with $0 \le k \le 3$ such
that there is a chain of~$k+1$ edges $e_1,e_2,\dots,e_{k+1}$ incident
to~$c$ such that~$e_1$ depends on~$e_2$ and $\dots$ and~$e_k$ depends
on~$e_{k+1}$, but there is no such chain of $k+2$ edges.
For example, in Figs.~\ref{fig:assignment_case_1},
\ref{fig:assignment_case_2}, and~\ref{fig:assignment_case_3}, the
assignment has a dependency depth of~1, whereas in
Fig.~\ref{fig:assignment_case_4}, the assignment has a dependency
depth of~0.
Showing that there is a valid assignment of dependency depth at most~1 will imply the existence of an appropriate set of grid points for the bend points as formalized in Lemmas~\ref{lem:2-bend-c-q-contains-grid-points} and~\ref{lem:2-bend-c-q-contains-grid-points-also-for-hierarchie-depth-2}.
In fact, as we will see in the discussion below, if we could avoid
dependencies, our drawing would fit on a grid of size
$\mathcal{O}(n^2) \times \mathcal{O}(n^2)$.
Unfortunately, with our current approach this seems to be unavoidable.

We now construct an assignment that we will show in
Lemma~\ref{lem:validAssignment} to be valid and to have dependency
depth at most~1.  The four cases of our assignment are given in
order of priority.  Note that, in Cases~1
and~2, our assignment always contains dependencies; see
Figs.~\ref{fig:assignment_case_1} and~\ref{fig:assignment_case_2}.
Note further that it is enough to specify the assignment of one edge;
the remaining assignment is determined since the circular orders of
the edges and the assigned half-lines must be the~same.

\begin{figure}[t]
	\begin{subfigure}[t]{0.24 \linewidth}
		\centering
		\includegraphics[page=1]{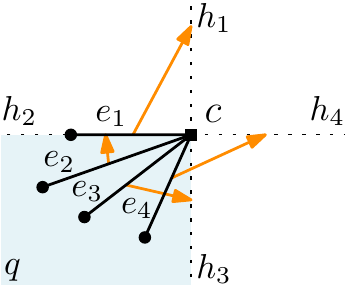}
		\caption{Case~1: $q$ contains four incident edges.}
		\label{fig:assignment_case_1}
	\end{subfigure}
	\hfill
	\begin{subfigure}[t]{0.24 \linewidth}
		\centering
		\includegraphics[page=1]{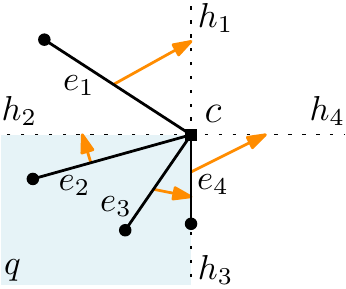}
		\caption{Case~2: $q$ contains three incident edges.}
		\label{fig:assignment_case_2}
	\end{subfigure}
	\hfill
	\begin{subfigure}[t]{0.24 \linewidth}
		\centering
		\includegraphics[page=1]{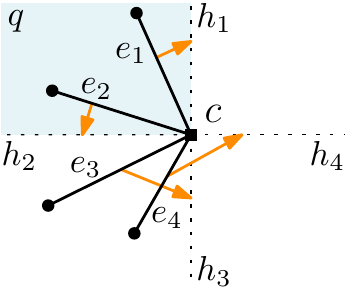}
		\caption{Case~3: $q$ contains two incident edges.}
		\label{fig:assignment_case_3}
	\end{subfigure}
	\hfill
	\begin{subfigure}[t]{0.24 \linewidth}
		\centering
		\includegraphics[page=1]{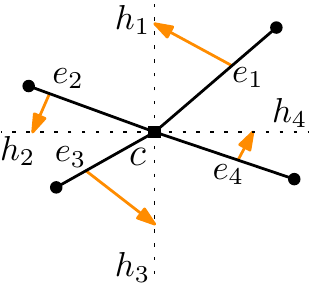}
		\caption{Case~4: One incident edge per quadrant.}
		\label{fig:assignment_case_4}
	\end{subfigure}
	
	\bigskip
	
	\begin{subfigure}[t]{0.24 \linewidth}
		\centering
		\includegraphics[page=2]{assignment_case_1}
		\caption{Case~1.}
		\label{fig:assignment_case_1-validity}
	\end{subfigure}
	\hfill
	\begin{subfigure}[t]{0.24 \linewidth}
		\centering
		\includegraphics[page=2]{assignment_case_2}
		\caption{Case~2.}
		\label{fig:assignment_case_2-validity}
	\end{subfigure}
	\hfill
	\begin{subfigure}[t]{0.24 \linewidth}
		\centering
		\includegraphics[page=2]{assignment_case_3}
		\caption{Case~3.}
		\label{fig:assignment_case_3-validity}
	\end{subfigure}
	\hfill
	\begin{subfigure}[t]{0.24 \linewidth}
		\centering
		\includegraphics[page=2]{assignment_case_4}
		\caption{Case~4.}
		\label{fig:assignment_case_4-validity}
	\end{subfigure}
	
	\caption{The four cases of our assignment procedure: (a)--(d)
		indicate the assignment with orange arrows and show that the
		dependency depth is always at most~1, (e)--(f) show that the
		assignment is valid; the radius of the light blue disk
		is~$\epsilon$.}
\end{figure}

\noindent
\begin{description}
	\item[Case~1:] There is a quadrant~$q$ that contains all four incident edges; see Fig.~\ref{fig:assignment_case_1}.
	
	Take the two ``inner'' edges in~$q$ and assign them to the two
	half-lines that bound~$q$ while keeping the circular order.
	
	\item[Case~2:] There is a quadrant~$q$ that contains three incident edges; see Fig.~\ref{fig:assignment_case_2}.
	
	Consider the edge outside~$q$, say~$e_1$, and assign it to the
	closest half-line~$h_i$ in terms of angular distance that does not bound~$q$.
	
	\item[Case~3:] There is a quadrant~$q$ that contains two incident edges; see Fig.~\ref{fig:assignment_case_3}.
	
	Assign the incident edges in~$q$ to the half-lines that bound~$q$ while keeping the circular order.
	
	\item[Case~4:] Each quadrant contains exactly one incident edge; see Fig.~\ref{fig:assignment_case_4}.
	
	Assign each edge to its closest half-line in counter-clockwise direction.
\end{description}

\newcounter{lemmaCounterValidAssignment}
\setcounter{lemmaCounterValidAssignment}{\thetheorem} 
\newcommand{\lemmaValidAssignment}{%
	Our assignment procedure returns a valid assignment with
	dependency depth at most~1.
}
\begin{lemma}
	\label{lem:validAssignment}
	\lemmaValidAssignment
\end{lemma}

\begin{proof}
	Observe that there is a disk with radius $\epsilon > 0$
	centered at~$c$ such that for every point~$p$ in this disk,
	the four line segments $\overline{a_1 p}, \overline{a_2 p},
	\overline{a_3 p}, \overline{a_4 p}$ do not cross the
	boundary of the subdivided kite.  In particular, by redrawing
	edges with bend points in this disk, we need only to worry
	about crossings among the edges incident to~$c$, not with
	edges of the kite.  To establish the lemma, it suffices to
	consider the four cases of our assignment independently.
	
	In Figs.~\ref{fig:assignment_case_1}--\ref{fig:assignment_case_4}
	the dependency depth is at most~1 in any of the four cases.  
	Note that only in Case~3 other configurations
	regarding the positions of $e_3$ and $e_4$ are possible, for
	example, when $e_3$ and $e_4$ lie in distinct quadrants or
	when~$e_3$ and~$e_4$ lie in the quadrant opposite~$q$.
	These alternate configurations result in all of $e_1$,
	$e_2$, $e_3$ and $e_4$ being independent.  Thus, we conclude
	that the dependency depth is always at most~1.
	
	Now, we place the bend points $b_1, \dots, b_4$ onto $A(e_1), \dots, A(e_4)$, respectively.
	For $i=1,\dots,4$, we determine the distance~$\epsilon_i$
	of~$b_i$ from~$c$, as follows.  If edge~$e_i$ is independent,
	we simply set $\epsilon_i=\epsilon$.
	Otherwise, if $e_i$ depends on~$e_j$, we first place~$b_j$,
	compute the intersection point~$x$ of $\overline{a_jb_j}$ with
	$A(e_i)$, and set $\epsilon_i = \|\overline{xc}\|/2$.  By this simple
	rule and the choice of~$\epsilon$ it is clear that no two
	redrawn edges intersect.  Hence, the assignment is valid.
\end{proof}

Note that Lemma~\ref{lem:validAssignment} already gives us a RAC$_2$
drawing of the input graph, but in order to get a (good) bound on the
grid size of the drawing, we have to place the bend points on a grid
that is as coarse as possible, but still fine enough to provide us
with grid points where we need them: on the half-lines emanating from
the crossing vertices.  This is what the remainder of this section is
about.

\begin{figure}[t]
	\begin{subfigure}[t]{0.33\linewidth}
		\centering
		\includegraphics[page=1] {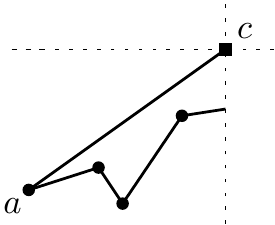}
		\caption{available polygon}
		\label{fig:polygon_of_available_area}
	\end{subfigure}
	\hfill
	\begin{subfigure}[t]{0.65\linewidth}
		\centering
		\includegraphics[page=2] {polygon_and_triangle}
		\caption{triangle for valid edge placement given
			points~$p$ and~$q$}
		\label{fig:valid_triangle}
	\end{subfigure}
	
	\caption{Example of an available polygon in which we determine the points~$p$ and~$q$ and with them the triangle for valid edge placement and the line segment~$\overline{q c}$.}
\end{figure}

\paragraph{Placement of Bend Points on the Grid}
In $\Gamma'$, we have a drawing of a subdivided kite for every crossing in the 1-plane input graph.
It is an octagon with a central crossing vertex~$c$ of degree four in its interior.
For an example, see Fig.~\ref{fig:inserting-bend-point-on-the-grid(c)}.
We will redraw the straight-line edges between~$c$ and its four
adjacent vertices as 1-bend edges according to the assignment~$A$ computed in the previous step.
The segment of such a 1-bend edge~$ac$ that ends at~$c$ will lie on the
axis-parallel half-line~$A(ac)$.
If we pair and concatenate the 1-bend edges that enter~$c$ from
opposite sides, we obtain two 2-bend edges and a right-angle crossing
in~$c$; see Fig.~\ref{fig:inserting-bend-point-on-the-grid(h)}.
It remains to show how the bend points for the edges are placed on the grid.
We proceed as follows. 

\begin{figure}[tb]
	\begin{subfigure}[t]{0.32 \linewidth}
		\centering
		\includegraphics[page=3] {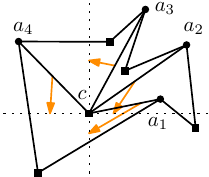}
		\caption{A subdivided kite obtained from the drawing step. The assignment of edges to half-lines is indicated by orange arrows.}
		\label{fig:inserting-bend-point-on-the-grid(c)}
	\end{subfigure}
	\hfill
	\begin{subfigure}[t]{0.32 \linewidth}
		\centering
		\includegraphics[page=4] {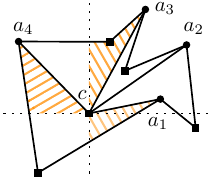}
		\caption{Available polygons for each pair of independent edge and assigned half-line.}
		\label{fig:inserting-bend-point-on-the-grid(d)}
	\end{subfigure}
	\hfill
	\begin{subfigure}[t]{0.32 \linewidth}
		\centering
		\includegraphics[page=5] {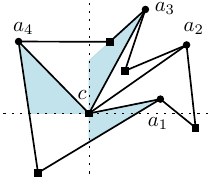}
		\caption{Triangles for valid edge placement of the three independent edges.}
		\label{fig:inserting-bend-point-on-the-grid(e)}
	\end{subfigure}
	
	\bigskip
	
	\begin{subfigure}[t]{0.32 \linewidth}
		\centering
		\includegraphics[page=6] {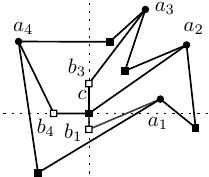}
		\caption{After the insertion of the bend points of the three independent edges.}
		\label{fig:inserting-bend-point-on-the-grid(f)}
	\end{subfigure}
	\hfill
	\begin{subfigure}[t]{0.32 \linewidth}
		\centering
		\includegraphics[page=7] {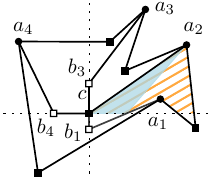}
		\caption{Available polygon and triangle for valid edge placement for the edge~$a_2c$ which depends on~$a_1c$.}
		\label{fig:inserting-bend-point-on-the-grid(g)}
	\end{subfigure}
	\hfill
	\begin{subfigure}[t]{0.32 \linewidth}
		\centering
		\includegraphics[page=8] {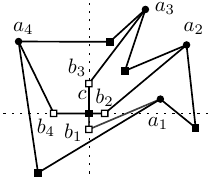}
		\caption{Result after the insertion of the bend point~$b_2$.}
		\label{fig:inserting-bend-point-on-the-grid(h)}
	\end{subfigure}
	
	\caption{Transformation from 
		a planarized crossing to a RAC$_2$ crossing.}
\end{figure}

First, we determine for each independent edge~$ac$ incident to a crossing vertex~$c$ the available region into which we can redraw~$ac$ with a bend~$b$ on~$A(ac)$.
The region between~$\overline{ac}$ and the half-line~$A(ac)$ inside the subdivided kite defines an \emph{available polygon}.
Examples of such an available polygon 
are given in
Figs.~\ref{fig:polygon_of_available_area}
and~\ref{fig:inserting-bend-point-on-the-grid(d)}.
Observe that there is only a triangle inside each available polygon in which the new line segment~$\overline{a b}$ can be placed.
Such a \emph{triangle for valid edge placement} is determined by~$a$,~$c$ and a corner point~$p$ of the available polygon.
The point~$p$ is the corner point (excluding~$a$ and~$c$) for which the angle between $\overline{a c}$ and $\overline{a p}$ inside the available polygon is the smallest.
These triangles for valid edge placement are depicted in
Figs.~\ref{fig:valid_triangle}
and~\ref{fig:inserting-bend-point-on-the-grid(e)}.
Observe that in such a triangle, the angle at~$a$ cannot become arbitrarily small because every determining point lies on a grid point.
Let~$q$ be the intersection point of the line through~$\overline{a p}$ and the half-line~$A(ac)$.
One can see~$q$ as the projection of~$p$ onto~$A(ac)$ seen from~$a$.
Note that we have a degenerate case if~$a \in A(ac)$.
Then, the available polygon has no area and equals the line segment~$\overline{a c}$.
In this case let $a = p = q$.
Moreover, note that $p$ can be equal to~$q$ because the intersection of $A(ac)$ and an edge of the subdivided kite is also a corner point of the available polygon.
This is the only case where $p$ may not be a grid point.

We will place the bend point~$b$ onto the line segment~$\overline{q c}$.
Observe that for a triangle for valid edge placement of an edge~$e_1$ that depends on another edge~$e_2$ in~$A$, $\| \overline{q c} \|$ would initially be equal to 0 since $q = c$ then.
For this reason, we first redraw the independent edges, which gives us some space for the edges depending on them, then compute the available polygons and the triangles for
valid edge placement for the other edges, and finally redraw those
edges.  Remember that our assignment procedure returns only
assignments with dependency depth at most~1.  Let~$\Gamma'$ be drawn
on a grid of size $\tilde{n} \times \tilde{n}$.  We refine the grid by
a factor of~$\tilde{n}$ in each dimension.  The next step in our
algorithm relies on the following lemma.  Recall
that we denote the x- and y-coordinate of a point~$p$ by~$x(p)$
and~$y(p)$, respectively.

An important tool in our analysis will be
the so-called \emph{Farey sequence}~\cite{wiki:Farey} of order
$\tilde{n}-1$.
This is the ordered set of all fractions~$f_1, f_2, \dots$, such that $f_i = \frac{a_i}{b_i}$, $a_i, b_i \in \mathbb{N}$, $a_i \leq b_i \leq \tilde{n}-1$, the greatest common divisor of $a_i$ and $b_i$ is~1, and $\frac{a_i}{b_i} < \frac{a_j}{b_j}$ for $i < j$.
In other words, it is the sequence of all reduced fractions from~0 to~1 with numerator and denominator less or equal to~$\tilde{n}-1$.
A nice property of neighboring numbers $f_i = \frac{a}{b}$
and $f_{i+1} = \frac{c}{d}$ in a Farey sequence is that
\begin{equation}
\label{eq:proof_grid_point_on_segment_4}
\frac{c}{d} - \frac{a}{b} = \frac{1}{b d}.
\end{equation}

\newcounter{lemmaCounterGridPointOnCQFirstRefinement}
\setcounter{lemmaCounterGridPointOnCQFirstRefinement}{\thetheorem} 
\newcommand{\lemmaGridPointOnCQFirstRefinement}{%
	For any independent edge~$ac$, the
	interior of the line segment~$\overline{q c}$ contains at least one grid point of the refined $\tilde{n}^2 \times \tilde{n}^2$ grid.
}
\begin{lemma}
	\label{lem:2-bend-c-q-contains-grid-points}
	\lemmaGridPointOnCQFirstRefinement
\end{lemma}
\begin{figure}[tb]
	\begin{subfigure}[t]{0.32 \linewidth}
		\centering
		\includegraphics{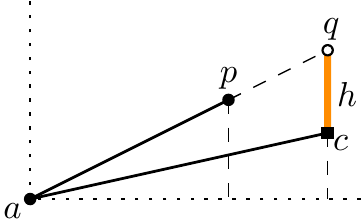}
		\caption{Case A1.}
		\label{fig:proof_grid_point_case_1}
	\end{subfigure}
	\hfill
	\begin{subfigure}[t]{0.32 \linewidth}
		\centering
		\includegraphics{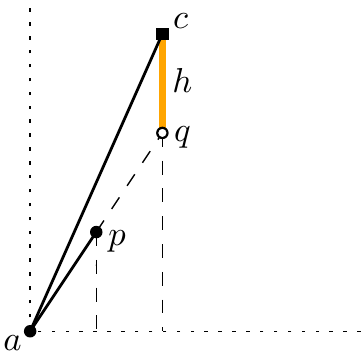}
		\caption{Case A4.}
		\label{fig:proof_grid_point_case_3}
	\end{subfigure}
	\hfill
	\begin{subfigure}[t]{0.32 \linewidth}
		\centering
		\includegraphics{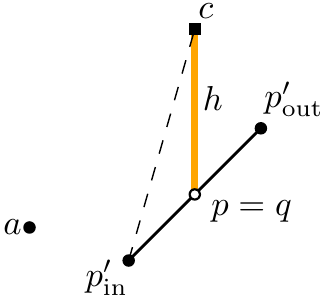}
		\caption{Case B.}
		\label{fig:proof_grid_point_case_b}
	\end{subfigure}
	\caption{Different cases concerning the analysis of
		$h=\|\overline{qc}\|$ in the proof of
		Lemma~\ref{lem:2-bend-c-q-contains-grid-points}.}
\end{figure}

\begin{proof}
	Without loss of generality, we can assume that
	$\overline{q c}$ is vertical. 
	If $x(a)=x(q)=x(c)$, we have the degenerate case
	$q = a$.
	We do not need to bend the edge $ac$ in our algorithm, but, for the completeness of the proof, we can easily see that there are at least $\tilde{n}-1$ grid points on the refined $\tilde{n}^2 \times \tilde{n}^2$ grid because $c$ and $q = a$ are grid points of the coarser $\tilde{n} \times \tilde{n}$ grid.
	
	So, without loss of generality, we can assume that
	$x(a) < x(q) = x(c)$, because mirroring the drawing
	with respect to the line through~$\overline{q c}$ does
	not change the structure of the drawing.
	We can also assume that $a=(0,0)$.
	Again, without loss of generality, we can assume that $y(c)
	\geq 0$. 
	If $y(c) = 0$, we can furthermore assume $y(q) > 0$ (both by the argument of mirroring across the $x$-axis).
	If $y(c) > 0$ and $y(q) < 0$, we are fine because $c$ and $(x(c), 0)$ are both  grid points of the coarser grid.
	Between them, there is more than one grid point of the finer grid.
	So we continue with $y(c) \geq 0$ and $y(q) \geq 0$.
	
	For convenience, we will work with coordinates on the coarser $\mathcal{O}(\tilde{n}) \times \mathcal{O}(\tilde{n})$ grid in the following case distinction.
	Moreover, observe that~$c$ does not lie on the top- or
	bottommost row or on the left- or rightmost column of
	the grid since~$c$ is enclosed by the dummy edges of a
	divided quadrangle.
	Therefore, we know that the difference in the $x$- and
	in the $y$-coordinate of~$c$ and any other vertex of the
	drawing is less than~$\tilde{n}$.
	In particular, we know that
	\begin{equation*}
	0 = x(a) < x(p) \leq  x(q) = x(c) < \tilde{n}.
	\end{equation*}
	
	Now, we distinguish two cases.
	
	\begin{description}
		\item[Case~A:] The point $p$ is a grid point.
		
		The points $a$, $p$ and $q$ are collinear.
		For $x(p) \geq y(p)$ and $x(c) \geq y(c)$, the slopes of $\overline{a p}$ and $\overline{a c}$ are values of the Farey sequence of order $\tilde{n}-1$.
		The slopes are ${y(p)}/{x(p)}$ and ${y(c)}/{x(c)}$.
		One can imagine all these possible slopes going out from~$a$ as rays.
		Without loss of generality, we can assume that the reduced fractions of ${y(p)}/{x(p)}$ and ${y(c)}/{x(c)}$ (or their reciprocals) are neighbored fractions in the Farey sequence and neighbored rays in the picture of the rays going out from~$a$.
		We also assume that ${y(p)}/{x(p)}$ and ${y(c)}/{x(c)}$ are reduced fractions because for a multiple of one of the Farey numbers, the line segment $\overline{q c}$ could only be longer and have more grid points of the finer grid on it but not fewer.
		
		We distinguish the following four subcases.
		
		\begin{description}
			\item[Case~A1:] $y(q) \geq y(c)$, and
			$\frac{y(p)}{x(p)}$ and $\frac{y(c)}{x(c)}$ are
			neighbors in the Farey sequence (see
			Fig.~\ref{fig:proof_grid_point_case_1}).
			
			We have
			\begin{equation}
			\label{eq:proof_grid_point_on_segment_1}
			h = \| \overline{q c} \| = y(q) - y(c)
			\end{equation}
			and
			\begin{equation*}
			y(q) = \frac{y(p)}{x(p)} \cdot x(c).
			\end{equation*}
			Putting this together, we get
			\begin{equation}
			\label{eq:proof_grid_point_on_segment_3}
			h = \frac{y(p)}{x(p)} \cdot x(c) - y(c) = x(c) \cdot
			\left( \frac{y(p)}{x(p)} - \frac{y(c)}{x(c)} \right).
			\end{equation}
			Due to $y(q) \geq y(c)$,
			we know that $\frac{y(p)}{x(p)} > \frac{y(c)}{x(c)}$.
			Using this and Equation~\ref{eq:proof_grid_point_on_segment_4} leads to
			\begin{equation}
			\label{eq:proof_grid_point_on_segment_5}
			h = x(c) \cdot \left( \frac{y(p)}{x(p)} - \frac{y(c)}{x(c)} \right) = x(c) \cdot \frac{1}{x(c) \cdot x(p)} = \frac{1}{x(p)} > \frac{1}{\tilde{n}}.
			\end{equation}
			
			\item[Case~A2:] $y(q) \leq y(c)$, and $\frac{y(p)}{x(p)}$ and $\frac{y(c)}{x(c)}$ are neighbors in the Farey sequence.
			
			This is almost the same as Case~A1,
			only multiplied with $-1$ because now
			we have $\frac{y(p)}{x(p)} < \frac{y(c)}{x(c)}$.
			Indeed, we have
			\begin{align*}
			h &= y(c) - y(q) \\
			&= y(c) - \frac{y(p)}{x(p)} \cdot x(c) \\
			&= x(c) \cdot \left( \frac{y(c)}{x(c)} - \frac{y(p)}{x(p)} \right) \\
			&= x(c) \cdot \frac{1}{x(c) \cdot x(p)} = \frac{1}{x(p)} > \frac{1}{\tilde{n}}.
			\end{align*}
			
			\item[Case~A3:] $y(q) \geq y(c)$, and $\frac{y(p)}{x(p)}$ and $\frac{y(c)}{x(c)}$ are not numbers of the Farey sequence because their numerator is greater than their denominator, but they can be seen as part of
			an extension of the Farey sequence
			from $1$ to $+\infty$. 
			Their reciprocals are neighbors in the Farey sequence.
			
			This case is also similar to A1.
			Equations~\ref{eq:proof_grid_point_on_segment_1}
			and~\ref{eq:proof_grid_point_on_segment_3}
			still hold, but we need to be careful with
			Equation~\ref{eq:proof_grid_point_on_segment_5}
			because $\frac{y(p)}{x(p)}$ and
			$\frac{y(c)}{x(c)}$ are not numbers of
			the Farey sequence. 
			An implication of
			Equation~\ref{eq:proof_grid_point_on_segment_4}
			is that 
			\begin{equation*}
			\frac{c}{d} - \frac{a}{b} =
			\frac{bc-ad}{bd} = \frac{1}{b d},
			\text{ which implies } bc-ad = 1.
			\end{equation*}
			Plugging in the Farey numbers $\frac{x(p)}{y(p)}$ and
			$\frac{x(c)}{y(c)}$ with $\frac{x(p)}{y(p)} <
			\frac{x(c)}{y(c)}$, we get
			\begin{equation*}
			y(p) \cdot x(c) - x(p) \cdot y(c) = 1 .
			\end{equation*}
			Using this, we transform
			Equation~\ref{eq:proof_grid_point_on_segment_3}, which
			yields the desired lower bound on~$h$:
			\begin{align*}
			h &= x(c) \cdot \left( \frac{y(p)}{x(p)} - \frac{y(c)}{x(c)} \right) \\
			&= x(c) \cdot \frac{y(p) \cdot x(c) - x(p) \cdot y(c)}{x(c) \cdot x(p)} \\
			&= x(c) \cdot \frac{1}{x(c) \cdot x(p)}
			= \frac{1}{x(p)} > \frac{1}{\tilde{n}}
			\end{align*}
			
			\item[Case~A4:] $y(q) \leq y(c)$, and $\frac{y(p)}{x(p)}$ and $\frac{y(c)}{x(c)}$ are not numbers of the Farey sequence because their numerator is greater than their denominator, but they can be seen as part of
			an extension of the Farey sequence
			from $1$ to $+\infty$. Their reciprocals are neighbors in the Farey sequence. A sketch is given in Fig.~\ref{fig:proof_grid_point_case_3}.
			
			This case is analogous to Case~A3 in
			the same way as Case~A2 is analogous
			to Case~A1.  Again, we can multiply
			with~$-1$ or alternatively swap all
			occurrences of~$p$ and~$c$.
		\end{description}
		
		\item[Case~B:] The point~$p$ is not a grid point.
		
		This situation may only occur if $p = q$.
		In this case the point~$p$ in the
		available polygon
		is the intersection of the
		assigned axis-parallel half-line and an
		edge~$e$ of the subdivided kite.
		We name the endpoint
		of~$e$ that is inside the available
		polygon~$p'_\mathrm{in}$ and the endpoint
		that is outside~$p'_\mathrm{out}$ (see
		Fig.~\ref{fig:proof_grid_point_case_b}).
		Clearly, we have a similar situation as in
		Case~A.  Here, $p'_\mathrm{in}$ is in the
		position of~$a$ in Case~A and
		$p'_\mathrm{out}$ is in a similar position
		as~$p$ in Case~A.  The points~$p'_\mathrm{in}$
		and~$p'_\mathrm{out}$ are vertices of~$G'$
		and, thus, grid points of the $\tilde{n}
		\times \tilde{n}$ grid.  The only difference
		is the order of the points $a$, $p$, $q$ and
		$p'_\mathrm{in}$, $q$, $p'_\mathrm{out}$ on
		each common line.  Observe that the formulas
		given in Case~A still hold if $q$ lies
		between~$p'_\mathrm{in}$
		and~$p'_\mathrm{out}$ instead of lying to
		the right of both.  Therefore, by doing the
		same analysis as in Case~A with exchanged
		roles of~$a$ and~$p$, we get the same result.
	\end{description}
	To summarize, for both cases and each subcase, we have
	seen that $\|\overline{q c}\| > {1}/{\tilde{n}}$.  By
	refining the $\tilde{n} \times \tilde{n}$ grid by a
	factor of $\tilde{n}$ in each dimension, we get a
	$\tilde{n}^2 \times \tilde{n}^2$ grid where each grid
	point of the coarser grid is also a grid point of the
	finer grid.  The crossing point~$c$ is a grid point of
	both grids.  On each of the four axis-parallel
	half-lines emanating from~$c$, we reach the next grid
	point after a distance of ${1}/{\tilde{n}}$.  Given
	that $\|\overline{q c}\| > {1}/{\tilde{n}}$, the
	interior of the line segment $\overline{q c}$ contains
	at least one grid point.
\end{proof}
Using Lemma~\ref{lem:2-bend-c-q-contains-grid-points}, we pick for
each independent edge any grid point of $\overline{q c}$, place a bend
point $b$ on it, and replace the segment $\overline{a c}$ by the two
segments $\overline{a b}$ and $\overline{b c}$.
An example is given in Fig.~\ref{fig:inserting-bend-point-on-the-grid(f)}, where the edges $a_1c$, $a_3c$, and $a_4c$ are independent, but $a_2c$ depends on~$a_1c$.

We again refine the grid by a factor of $\tilde{n}$ in each dimension. The grid size is now $\tilde{n}^3 \times \tilde{n}^3$.
For the remaining edges incident to a crossing vertex~$c$, we compute new available polygons and triangles for valid edge placement since we need to take the 1-bend edges into account that were inserted in the previous step.
Now the following lemma yields grid points for the bend points of the remaining edges.
\newcounter{lemmaCounterGridPointOnCQSecondRefinement}
\setcounter{lemmaCounterGridPointOnCQSecondRefinement}{\thetheorem} 
\newcommand{\lemmaGridPointOnCQSecondRefinement}{%
	After having redrawn the independent edges, the interior of
	the line segment~$\overline{q c}$ of each edge~$ac$ depending on an independent edge~$\hat{a}c$ contains at least one grid
	point of the refined $\tilde{n}^3 \times \tilde{n}^3$
	grid.
}
\begin{lemma}
	\label{lem:2-bend-c-q-contains-grid-points-also-for-hierarchie-depth-2}
	\lemmaGridPointOnCQSecondRefinement
\end{lemma}

\begin{figure}[t]
	\centering
	\includegraphics[] {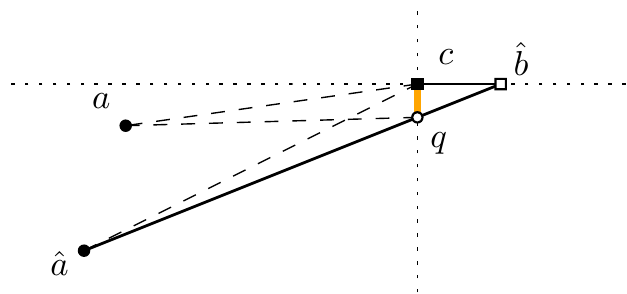}
	\caption{The bend point of the edge $ac$ is placed on a grid point of the interior of the line segment $\overline{c q}$. The point~$q$ depends on the placement of $\hat{b}$, which is the bend point of the edge $\hat{a}c$.}
	\label{fig:proof_dependent_bend_point}
\end{figure}

\begin{proof}
	All of the following coordinates are relative to the grid of size $\tilde{n}^2 \times \tilde{n}^2$ that has been refined once.
	We assume that~$\hat{a}c$ has been bent on $\hat{b}$ in the previous step.
	Given Lemma~\ref{lem:2-bend-c-q-contains-grid-points},
	we have to consider only the case that $p = q \in \overline{\hat{a} \hat{b}}$.
	
	We assume, without loss of generality, that $x(\hat{a}) = 0$ and $y(\hat{a}) = 0$.
	Furthermore, we assume that $x(c) \geq 0$ and $y(c) \geq 0$ because mirroring across some axis-parallel line does not change the structure of the drawing.
	We assume, without loss of generality, that $\hat{b}$ lies on the half-line originating at $c$ and going to positive infinity in the $x$-dimension because $\overline{\hat{a} \hat{b}}$ crosses some other axis-parallel half-line (here: the one going to negative infinity in $y$-dimension) and, again, mirroring does not change the structure of the drawing.
	This implies $y(c) > y(\hat{a})$.
	Our current situation is depicted in Fig.~\ref{fig:proof_dependent_bend_point}.
	Now, we analyze how short the line segment $\overline{q c}$ can become in the worst case.		
	The line segment will become shorter if
	\begin{itemize}
		\item the $x$-distance $x(\hat{b}) - x(c)$ decreases or
		\item the $y$-distance $y(c) - y(\hat{a})$ decreases or
		\item the $x$-distance $x(c) - x(\hat{a})$ increases.
	\end{itemize}
	So $\overline{q c}$ will be shortest if we assume the
	most extremes of these values, namely
	\begin{itemize}
		\item $x$-distance $x(\hat{b}) - x(c) = 1$, and
		\item $y$-distance $y(c) - y(\hat{a}) = \tilde{n}$ (it cannot become smaller because both are points of the coarser $\tilde{n} \times \tilde{n}$ grid and $y(c) > y(\hat{a})$), and
		\item $x$-distance $x(c) - x(\hat{a}) = (\tilde{n} - 1) \cdot \tilde{n}$
		(This is because both are grid points on the coarser $\tilde{n} \times \tilde{n}$ grid.
		Since $\hat{b}$ is on the right side of both, they cannot both be outermost grid points and, thus, they can only have a distance of $\tilde{n} - 1$ on the initial coarser grid and $(\tilde{n} - 1) \cdot \tilde{n}$ on the $\tilde{n}^2 \times \tilde{n}^2$ grid.)
	\end{itemize}
	Hence, for the slope $m$ of $\overline{\hat{a}
		\hat{b}}$, we get
	\begin{equation*}
	m = \frac{\tilde{n}}{\tilde{n}^2 - \tilde{n} + 1}.
	\end{equation*}
	Using this, we determine $y(q)$ by
	\begin{equation*}
	y(q) = \left( \tilde{n}^2 - \tilde{n} \right)  \cdot m = \left( \tilde{n}^2 - \tilde{n} \right)  \cdot \frac{\tilde{n}}{\tilde{n}^2 - \tilde{n} + 1} = \frac{\tilde{n}^3 - \tilde{n}^2}{\tilde{n}^2 - \tilde{n} + 1}.
	\end{equation*}
	Now, we can compute the length of the line segment $\overline{q c}$ this way:
	\begin{align*}
	y(c) - y(q) &= \tilde{n} - \frac{\tilde{n}^3 - \tilde{n}^2}{\tilde{n}^2 - \tilde{n} + 1} = \frac{\tilde{n}^3 - \tilde{n}^2 + \tilde{n} - \tilde{n}^3 + \tilde{n}^2}{\tilde{n}^2 - \tilde{n} + 1} \\
	&= \frac{\tilde{n}}{\tilde{n}^2 - \tilde{n} + 1} = \frac{1}{\tilde{n} - 1 + \frac{1}{\tilde{n}}}
	> \frac{1}{\tilde{n}} \hspace{6 pt} (\text{for } \tilde{n} > 1)
	\end{align*}
	With the same argument as in the proof of
	Lemma~\ref{lem:2-bend-c-q-contains-grid-points}, we
	see that the interior of $\overline{q c}$ contains
	always at least one grid point of the refined
	$\tilde{n}^3 \times \tilde{n}^3$ grid.
\end{proof}

For each remaining edge incident to a crossing vertex~$c$, we pick any grid point of its line segment~$\overline{q c}$ and place a bend point~$b$ on it.
Again, we replace $\overline{a c}$ by the two line segments $\overline{a b}$ and $\overline{b c}$.

\paragraph{Result}
Finally, we remove the dummy edges and dummy vertices that bound the subdivided kites and interpret the crossing vertices as crossing points.
We return the resulting RAC$_2$ drawing~$\Gamma$.
Now we analyze the size of the grid that ``carries''~$\Gamma$.
\begin{lemma}
	Every vertex, bend point and crossing point of
	$\Gamma$ lies on a grid of size at most $(8 n'^3 - 48
	n'^2 + 96 n' - 64) \times (4 n'^3 - 24 n'^2 + 48 n' -
	32)$, where $n'=n+5\crn(\mathcal{E})$ and
	$\crn(\mathcal{E})$ is the number of crossings
	in~$\mathcal{E}$.
\end{lemma}
\begin{proof}
	In the preprocessing, we build subdivided kites around each crossing.
	To this end, we insert four 2-paths per crossing, which means that we insert four new vertices around each crossing.
	Moreover we make every crossing point a vertex.
	Thus, the resulting plane graph $(G', \mathcal{E}')$ has $n' = n + 5 \crn(\mathcal{E})$ vertices.
	
	The shift algorithm places every vertex of the plane graph $(G', \mathcal{E}')$ onto a grid point of a grid of size $(2 n' - 4) \times (n' - 2)$.
	We refine this twice by $\hat{n} = (2 n' - 4)$ and obtain a grid of size:
	\begin{align*}
	\text{total grid size} &=(2 n' - 4)^3 \times (n' - 2)(2 n' - 4)^2 \nonumber \\
	&= (8 n'^3 - 48 n'^2 + 96 n' - 64) \times (4 n'^3 - 24 n'^2 + 48 n' - 32)
	\end{align*}
\end{proof}
Note that $\crn(\mathcal{E}) \leq n-2$ for 1-plane graphs~\cite{Czap2013}.
If we ignore the bend points, the drawing is on a grid of size $(2 n'
- 4) \times (n'-2)$, i.e., its size is quadratic.
Again, the algorithm by Harel and Sardas~\cite{Harel1998} and our modification run in linear time.
Therefore, we conclude the correctness of Theorem~\ref{thm:1-planar2-bendRACQuadraticArea}.

\section{Preserving Embeddings}
In this section, we show how to preserve the embedding when we compute 1-planar RAC$_1$ and IC-planar RAC$_0$ drawings from 1-plane and straight-line drawable IC-plane graphs, respectively.
There are algorithms known that compute such drawings from 1-plane and IC-plane graphs, but they may change the input embedding.
We describe how to modify these algorithms so that the input embedding is preserved in the output.
This means that the two containment relations shown in the diagram in Fig.~\ref{fig:DiagramBendRACClassRelation} with canceled ``$\mathcal{E}?$'' also hold for fixed embeddings.

\subsection{1-Planar 1-Bend RAC Drawings}
\label{sec:1-planar1-bendRACSameEmbedding}
Bekos et al.~\cite{Bekos2017} describe an algorithm for computing 1-planar RAC$_1$ drawings of 1-planar graphs in linear time.
Their algorithm takes a 1-plane graph as input, but the embedding may be changed during the execution of the algorithm, i.e., 
while the output is indeed a drawing of the same graph, it can induce a different 1-planar embedding.
In fact, they explicitly ask if every 1-planar embedding admits a RAC$_1$ drawing.
We answer their question in the affirmative by describing how
to modify their algorithm; see Theorem~\ref{thm:1-planar1-bendRACSameEmbeddingCopy}. 

\begin{theorem}
	\label{thm:1-planar1-bendRACSameEmbeddingCopy}
	\theoremOnePlaneOneBendRAC
\end{theorem}

To establish this we describe the original algorithm and then our modifications thereof.

\paragraph{Original Algorithm}
The algorithm starts with an \emph{augmentation} step.
In the 1-plane input graph $(G, \mathcal{E})$, dummy edges are inserted around each pair of crossing edges to induce empty kites (empty kites are defined in Section~\ref{sec:NIC-planar1-bendRACQuadraticArea}).
During this process, parallel edges can occur, but no new crossings.
They remove the original edge from each set of parallel edges (this changes the embedding), and for each face of degree two, i.e., a face bounded by two parallel edges, they remove one of the edges.
There can still be parallel dummy edges.
At the end of the augmentation step they triangulate each face by inserting dummy edges and vertices to obtain a triangulated 1-plane multigraph $(G^{+}, \mathcal{E}^{+})$.

The next step is computing a \emph{hierarchical contraction} of $(G^{+}, \mathcal{E}^{+})$.
For each set of parallel edges there is an inner graph component separated from the rest of the graph by the two outermost edges of these parallel edges.
This inner component is contracted to a single \emph{thick edge}, to which the information about the contracted subgraph is saved.
This contraction operation is applied (recursively) to every set of parallel edges.
In this way, they obtain a hierarchy of simple 1-plane 3-connected triangulated graphs.
The top-level graph is denoted by $(G^{*}, \mathcal{E}^{*})$.

The last step of the algorithm is \emph{drawing} the graph.
They remove the crossing edges from $(G^{*}, \mathcal{E}^{*})$ and draw it with an algorithm that delivers strictly convex straight-line drawings where the outer face is a prescribed convex polygon.
The linear-time algorithm by Chiba et al.~\cite{Chiba1984} fulfills these requirements.
They pass, as the prescribed polygon, a trapezoid if the outer face has degree four\footnote{i.e., when a crossing on the outer face was removed at the beginning of the drawing step} and a triangle otherwise.
Next, they manually reinsert the crossing edges.
For the inner convex faces, they draw one edge straight-line and the other edge with a bend so that it crosses the first edge at a right angle.
For the outer faces, they bend both edges.
This procedure is applied recursively for each subgraph contracted to a thick edge.
Since they can prescribe the shape of the outer face, they can always pass a shape that fits into the free space next to a thick edge to expand each subgraph.
In the end they remove the dummy edges and vertices that have not been part of the input graph and obtain a 1-planar RAC$_1$ drawing of the input graph.
Note that the embedding may have changed during the execution of the augmentation step where they had parallel edges.

\paragraph{Our Modifications}
Our modification in the \emph{augmentation} is to keep the original edges that are not part of a crossing.
But like Bekos et al., for each original edge that crosses another edge, we remove it if it gained parallel edges during the augmentation step.
When we remove such a crossing edge $e$, an empty kite becomes a divided quadrangle (see Fig.~\ref{fig:removal-of-parallel-crossed-edge}; divided quadrangles are defined in Section~\ref{sec:NIC-planar1-bendRACQuadraticArea}).	
\begin{figure}[tb]
	\begin{subfigure}[t]{0.48 \linewidth}
		\centering
		\includegraphics[page=1]
		{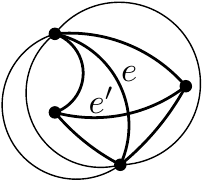}
		\caption{Empty kite with crossing edges $e$ and~$e'$. The original edge~$e$ gained parallel dummy edges in the execution of the augmentation step.}
	\end{subfigure}
	\hfill
	\begin{subfigure}[t]{0.48 \linewidth}
		\centering
		\includegraphics[page=2]
		{parallel_edges}
		\caption{After the removal of the original edge~$e$, the empty kite becomes a divided quadrangle containing $e'$ as inner edge.}
	\end{subfigure}
	
	\caption{Removal of an original edge that has parallel edges and is part of a crossing.}
	\label{fig:removal-of-parallel-crossed-edge}
\end{figure}
Suppose the edge~$e'$ crossed $e$ before $e$'s removal.
Note that the edge~$e'$ cannot have parallel dummy edges since these would cross either $e$ or a parallel dummy edge of~$e$, but, as stated earlier, no inserted dummy edge results in a new crossing and a crossing with $e$ would violate the 1-planarity.
We remember from where we removed these edges because we will reinsert them later.

We do not modify the \emph{hierarchical contraction} step, but we save the order of the subgraphs contracted at each separation pair to a thick edge and save the relative position of the original edge.

\begin{figure}[tb]
	\begin{subfigure}{0.48 \linewidth}
		\centering
		\includegraphics[page=1]
		{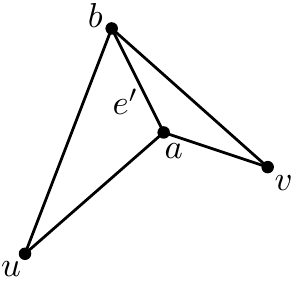}
		\caption{The divided quadrangle as returned by the shift algorithm.}
		\label{fig:reinstering-original-crossed-edge(a)}
	\end{subfigure}
	\hfill
	\begin{subfigure}{0.48 \linewidth}
		\centering
		\includegraphics[page=2]
		{reinserting_original_crossing}
		\caption{Thales' circles onto which we can place the later crossing point $c$.}
		\label{fig:reinstering-original-crossed-edge(b)}
	\end{subfigure}

            \bigskip

	\begin{subfigure}{0.48 \linewidth}
		\centering
		\includegraphics[page=3]
		{reinserting_original_crossing}
		\caption{Crossing point $c$ and the first parts of the crossing edges $e$ and $e'$ up to their bend points.}
		\label{fig:reinstering-original-crossed-edge(c)}
	\end{subfigure}
	\hfill
	\begin{subfigure}{0.48 \linewidth}
		\centering
		\includegraphics[page=4]
		{reinserting_original_crossing}
		\caption{Reinserted 1-bend edges $e$ and $e'$ that cross at a right angle inside the quadrangle $(a, v, b, u)$.}
		\label{fig:reinstering-original-crossed-edge(d)}
	\end{subfigure}
	
	\caption{Reinserting the original edge $e = uv$, which crosses the edge $e' = ab$ at a right angle.}
\end{figure}

The \emph{drawing} step is almost the same as in the original algorithm, but we make sure that we draw the inner subgraphs that were contracted at a separation pair~$\{u, v\}$ in the original order.
We distinguish two~cases.

\begin{description}
	\item[Case~1:] We have kept the original edge $uv$ from the set of parallel edges between $u$~and~$v$.
	
	In the original paper, they insert inner graphs stacked on one side (or they do not care on which side) of the straight-line segment~$\overline{u v}$.
	We insert the original edge $uv$ as straight-line segment and draw the subgraphs that have been to the left side of this edge in the original embedding on the left side of~$\overline{u v}$ in their original internal order.
	Analogously, we proceed with the subgraphs on the right side of~$\overline{u v}$.
	
	\item[Case~2:] We have removed the original edge $uv$ from the set of parallel edges between $u$~and~$v$.
	
	Again, we draw all subgraphs of the separation pair $\{u, v\}$ in their original internal order, but now we do not have the original edge $uv$ as the straight-line segment $\overline{u v}$.
	So we can draw them on one side or on both sides of~$\overline{u v}$.
	There will not be an edge at the straight-line segment~$\overline{u v}$.
	Instead, we reinsert the original edge with a bend at its original place in the embedding, i.e., into a divided quadrangle crossing an edge~$e'$, as follows.
	This divided quadrangle consists of two faces: one face has the two endpoints of~$e'$ (let these be $a$ and $b$) and $u$ as corner points, and the other face has $a, b, v$ as corner points (see Fig.~\ref{fig:reinstering-original-crossed-edge(a)}).
	To obtain a RAC$_1$ drawing, we reinsert $e$ in the following way.
	We remove the straight-line edge~$e'$ so that we obtain the empty quadrangle~$(a, v, b, u)$.
	We will choose a point $c$ on the Thales' circle around $\overline{a u}$ or $\overline{b u}$ that lies strictly inside the triangle $(a, b, u)$ (see Fig.~\ref{fig:reinstering-original-crossed-edge(b)}).
	To do this we first establish its existence.
	Assume for contradiction that it does not exist.
	Then, $b$ lies inside the Thales' circle around~$\overline{a u}$.
	Therefore, the angle~$\angle a b u$ is greater than 90 degrees.
	Analogously to $b$, $a$ must lie inside the Thales' circle around $\overline{b u}$.
	Therefore, the angle $\angle u a b$ is also greater than 90 degrees and the triangle~$(a, b, u)$ has a sum of internal angles that exceeds 180 degrees.
	This is a contradiction and, thus, there is a point on one of these two Thales' circles inside the triangle~$(a, b, u)$.
	Clearly, such a point~$c$ can be found in constant time.
	We will use $c$ as the crossing point.
	Without loss of generality, let $c$ lie on the Thales' circle around~$\overline{a u}$.
	We draw the first part of~$e$ and~$e'$ as straight-line segments~$\overline{u c}$ and~$\overline{a c}$, respectively.
	Now, we lengthen the segment of $e'$ over $c$ a little so that we are still inside~$(a, b, u)$; see Fig.~\ref{fig:reinstering-original-crossed-edge(c)}.
	From there we can reach $b$ with another straight-line segment because this vertex is a corner point of the triangle we are currently in and we have already passed~$e$.
	We also lengthen the straight-line segment of $e$ over $c$ until it reaches the other triangle, i.e., $(a, v, b)$, but does not pass or touch the border of the whole face of the empty quadrangle.
	From that point, we can reach~$v$ with another straight-line segment.
	These are our bend points for $e$ and $e'$ (see Fig.~\ref{fig:reinstering-original-crossed-edge(d)}).
\end{description}
After having removed the dummy edges and vertices, we obtain a drawing of the 1-plane input graph in its original embedding.
Observe that our modifications do not require more than linear time.
Like the original algorithm, the adapted version also only bends edges that participate in a crossing.
Edges that are not crossed are drawn as straight-line segments.

\subsection{IC-Planar Straight-Line RAC Drawings}
\label{sec:IC-planar0-bendRACSameEmbedding}
Brandenburg et al.~\cite{Brandenburg2016} describe an algorithm for computing IC-planar RAC$_0$ drawings of IC-planar graphs in cubic time.
Their algorithm takes an IC-plane graph as input, but this embedding may be changed during the execution of the algorithm, i.e., 
while the output is indeed a drawing of the same graph, it can induce a different IC-planar embedding.
We describe a slight modification to their algorithm to preserve the input embedding and obtain Theorem~\ref{thm:IC-planar0-bendRACSameEmbeddingCopy} below.

\begin{figure}[tb]
	\begin{subfigure}[t]{0.48 \linewidth}
		\centering
		\includegraphics[page=1]
		{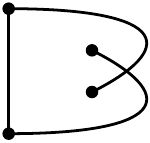}
		\caption{B-configuration}
		\label{fig:B-Config}
	\end{subfigure}
	\hfill
	\begin{subfigure}[t]{0.48 \linewidth}
		\centering
		\includegraphics[page=2]
		{b-w-configuration}
		\caption{W-configuration}
		\label{fig:W-Config}
	\end{subfigure}
	\caption{Embedded graph configurations of 1-plane graphs that cannot be drawn with straight-line edges.}
	\label{fig:B-andW-Config}
\end{figure}

Note that, as shown by Thomassen~\cite{Thomassen1988}, the straight-line drawable 1-plane graphs are precisely the 1-plane graphs without any so called \emph{B-} or \emph{W-configuration} (see Fig.~\ref{fig:B-andW-Config}).
Since IC-plane graphs cannot contain W-configurations, any IC-plane graph without B-configurations is straight-line drawable.
Clearly, our modifications only work for these straight-line drawable IC-plane graphs.

\begin{theorem}
	\label{thm:IC-planar0-bendRACSameEmbeddingCopy}
	\theoremICPlaneStraightLineRAC
\end{theorem}

To establish this we describe the original algorithm and then our modifications thereof.

\paragraph{Original Algorithm}
The algorithm starts with an \emph{augmentation step} to
    obtain a plane-maximal IC-plane graph $(G^{+},
    \mathcal{E}^{+})$ from the IC-plane input graph $(G,
    \mathcal{E})$, where a \emph{plane-maximal} IC-plane graph is
    an IC-plane graph to which no edge can be added without
    creating a new crossing.
In this step, Brandenburg et al.\ do not only add edges, but
    also re-route edges, which may change the embedding.  They
    guarantee that $(G^{+}, \mathcal{E}^{+})$ has the following
    three properties:
\begin{enumerate}[(P1)]
    \item \label{item:crossings-induce-empty-kites} The four
      endpoints of each pair of crossing edges induce an empty
      kite (for the definition of empty kites, refer to
      Section~\ref{sec:NIC-planar1-bendRACQuadraticArea}).
    \item \label{item:ic-plane-triangulated} After removing one
      edge of each pair of crossing edges, the resulting graph is
      plane and triangulated.
    \item \label{item:outer-face-is-3-cycle} The outer face is a
      3-cycle of non-crossing edges.
\end{enumerate}
To satisfy
    property~(P\ref{item:crossings-induce-empty-kites}),
    Brandenburg et al.\ add and re-route edges such that every
    crossing induces an empty kite.  In this way they also lose
    B-configurations, which are not straight-line drawable.  They
    satisfy property~(P\ref{item:ic-plane-triangulated}) by
    triangulating the remaining faces.  For
    property~(P\ref{item:outer-face-is-3-cycle}), they argue that
    the graph has a triangular face without crossing edges, which
    can be made the outer face by re-embedding the graph.

In the next step, the \emph{drawing step}, Brandenburg et al.\
    draw the maximal IC-plane graph $(G^{+}, \mathcal{E}^{+})$.
    To this end, they modify the shift algorithm of de Fraysseix
    et al.~\cite{Fraysseix1990} as follows.  From each pair of
    crossing edges, they temporarily remove one edge.  This yields
    a planar graph.  Then, Brandenburg et al.\ compute a canonical
    ordering of this graph.  When their incremental drawing
    procedure has placed all four vertices of a pair of crossing
    edges, they perform additional \emph{move} and \emph{lift}
    operations to make one of the crossing edges horizontal and
    the other one vertical.  As a consequence, all crossings are
    at right angles.

These move and lift operations preserve most invariants of the
    shift algorithm such as all edges on the outer face having
    slope~$\pm 1$.  The invariant concerning the grid size,
    however, is not preserved, and the grid size can be
    exponential in the number of vertices.

    Finally Brandenburg et al.\ remove the dummy edges, which they
    added during the augmentation step.  This yields an
    IC-planar RAC$_0$ drawing of the input graph $G$.  Note that,
    due to the changes in the augmentation step, the drawing does
    not necessarily conform to the input embedding $\mathcal{E}$.

\paragraph{Our Modifications}
We suggest some modifications to preserve the input
    embedding~$\mathcal{E}$.  Since IC-plane graphs with
    B-configurations are not straight-line drawable, we assume
    that our IC-plane input graph~$(G, \mathcal{E})$ does not
    contain any B-configuration.

\begin{figure}[tb]
	\begin{subfigure}[t]{0.28 \linewidth}
		\centering
		\includegraphics[page=1] {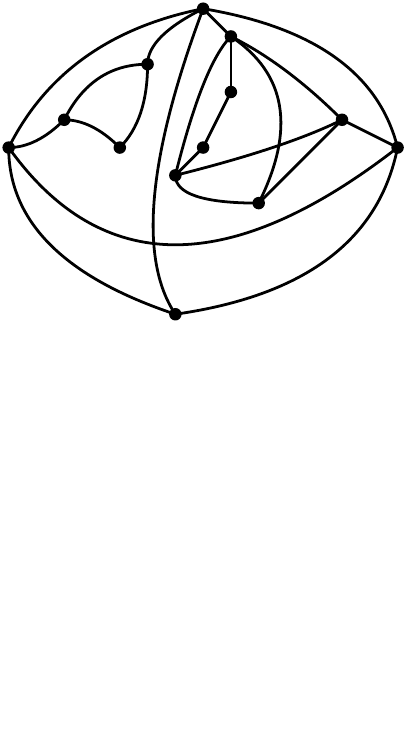}
		\caption{IC-plane graph where every crossing
                      is part of a (non-empty)
                      kite.}
		\label{fig:ic-plane_rac_hierarchical_construction(a)}
	\end{subfigure}
	\hfill
	\begin{subfigure}[t]{0.68 \linewidth}
		\centering
		\includegraphics[page=3] {ic-plane_rac_hierarchical_construction}
		\caption{Tree built from the graph
                      in~(a).  Every subgraph inside a kite has
                      been extracted to obtain empty kites.  Dummy
                      edges and vertices whose task it is to turn
                      the outer face of a subgraph into a triangle
                      are colored blue.}
		\label{fig:ic-plane_rac_hierarchical_construction(b)}
	\end{subfigure}
	\caption{Our modification for obtaining empty kites
              around each crossing.}
	\label{fig:ic-plane_rac_hierarchical_construction}
\end{figure}

In the \emph{augmentation} step, we do not obtain a single
    plane-maximal IC-plane graph with the previously specified
    properties, but a tree of plane-maximal IC-plane
    graphs, all of which fulfill these properties. 
We start by inserting dummy edges such that there is a (not
    necessarily empty) kite at every crossing.
Note: if an edge is already present in the graph, we do not create any copies of it, i.e., no parallel edges are introduced.
Instead of re-routing outer edges of a kite to make it empty,
    we extract, for each outer kite edge~$e$, the subgraph~$G_e$
    between~$e$ and the crossing edges.  To~$G_e$, we add a
    copy~$e'=u'v'$ of~$e=uv$ and connect the copied endpoints~$u'$
    and~$v'$ to a new vertex~$w'$.  Hence, the outer face of~$G_e$
    is a triangle, the triangle $u'v'w'$; see
    Fig.~\ref{fig:ic-plane_rac_hierarchical_construction} for an
    example.  After extracting its at most four subgraphs, the
    kite is empty.	We link each face of the now empty kite to the
    subgraph that has been there before.
We proceed recursively to ensure empty kites also in the
    subgraphs and obtain a tree of IC-plane graphs; see
    Fig.~\ref{fig:ic-plane_rac_hierarchical_construction(b)}.
In every graph of the tree, we triangulate the
    remaining faces in the same way as the original algorithm.
We do not have to re-embed the graph to fulfill
    property~(P\ref{item:outer-face-is-3-cycle}) because (i)~all
    faces have been triangulated and (ii)~edges on the outer face
    don't cross.  The latter is due to the fact that every crossing
    has been enclosed by an empty kite.

\begin{figure}[t!]
	\begin{subfigure}[t]{0.52 \linewidth}
		\centering
		\includegraphics[page=1]{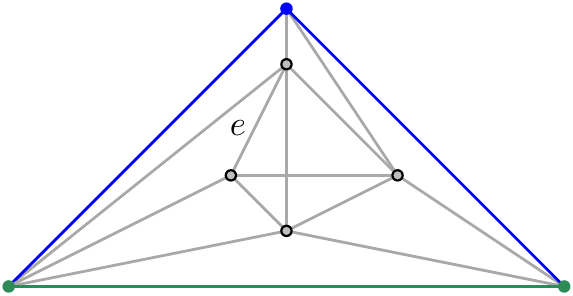}
		\caption{IC-planar RAC drawing of some level
                      with an empty kite enclosing a right-angle
                      crossing.
                    }
		\label{fig:ic-plane_rac_hierarchical_insertion(a)}
	\end{subfigure}
	\hfill
	\begin{subfigure}[t]{0.44 \linewidth}
		\centering
		\includegraphics[page=2]{ic-plane_rac_hierarchical_insertion}
		\caption{Fitting a scaled and rotated copy of
                      a drawing of~$G_e$ into a face of the empty
                      kite.}
		\label{fig:ic-plane_rac_hierarchical_insertion(b)}
	\end{subfigure}
	
            \bigskip

	\begin{subfigure}[t]{0.52 \linewidth}
		\centering
		\includegraphics[page=1]{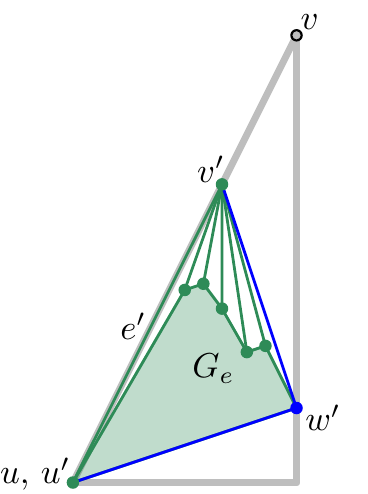}
		\caption{Detailed view of the face into which
                      we insert the IC-planar RAC drawing of the
                      subgraph~$G_e$.}
		\label{fig:ic-plane_rac_hierarchical_insertion(c)}
	\end{subfigure}
	\hfill
	\begin{subfigure}[t]{0.44 \linewidth}
		\centering
		\includegraphics[page=2]{ic-plane_rac_local_insertion}
		\caption{Result of moving~$v'$ to the
                      corresponding upper-level vertex~$v$.}
		\label{fig:ic-plane_rac_hierarchical_insertion(d)}
	\end{subfigure}
	
	\caption{Combining two drawings of subgraphs from
              neighboring levels of the tree.}
	\label{fig:ic-plane_rac_hierarchical_insertion}
\end{figure}

In the \emph{drawing} step, we produce an IC-planar RAC$_0$
    drawing of every graph in every node of the tree
    in the same way as Brandenburg et al.\ do for the graph as a whole.
Then, we combine the single drawings.
We start with the drawing of the top-level graph.
As illustrated by the example in
    Figs.~\ref{fig:ic-plane_rac_hierarchical_insertion(a)}--\ref{fig:ic-plane_rac_hierarchical_insertion(b)},
for every face bounded by an edge~$e$ and two crossing edges
    that is linked to a subgraph~$G_e$, we take the drawing
    of this subgraph, and rotate and scale it as described below.  
    Clearly, rotating and scaling do not change any internal
    angle.  We rotate and scale the drawing of the subgraph~$G_e$
    such that (i)~it fits into the drawing of the corresponding
    face of the upper-level graph, (ii)~the copy~$e'=u'v'$ of~$e=uv$
    in~$G_e$ lies on~$e$ and one of the endpoints of~$e'$ lies on
    the corresponding endpoint of~$e$, say, $u'$ lies on $u$, and
    (iii)~the dummy vertex~$w'$ (that we connected to~$u'$ and~$v'$
    when building the tree) lies on one of the two crossing
    edges of the higher level.
Note that this is not yet a correct drawing of the combined
    subgraphs: usually~$e'$ will be shorter than~$e$
    (see the example in
    Fig.~\ref{fig:ic-plane_rac_hierarchical_insertion}).
To make~$e'$ as long as~$e$, we simply move~$v'$ to~$v$.
    (Recall that we assumed that $u'$ already lies on~$u$.)
Note that this operation does not cause any new crossings because
    it is equivalent to a shift operation as
    performed by the shift algorithm.  Moreover, it cannot violate
    a right-angle crossing in the subgraph since the whole graph
    is IC-plane and, therefore, the shifted vertex is not incident
    to a crossing edge in the subgraph (only in the higher-level
    graph).  An example of this step is given in
    Figs.~\ref{fig:ic-plane_rac_hierarchical_insertion(c)}--\ref{fig:ic-plane_rac_hierarchical_insertion(d)}.

We proceed recursively with the subgraphs of each inserted subgraph.
After having removed all dummy vertices and edges (which,
    hence, don't cause any crossings), we obtain
    an IC-planar RAC$_0$ drawing of the input graph~$G$ in the
    IC-planar input embedding~$\mathcal{E}$.
Our modifications can be performed in linear time.
Therefore, the modified algorithm runs in cubic time,
    matching the running time of the original algorithm.

\section{Conclusion and Open Questions}

The two main results that we have presented in this paper concern
compact RAC drawings of 1-planar graphs with a low number of bends
per edge.  More precisely, we have shown that any $n$-vertex NIC-plane
graph admits a RAC$_1^\text{poly}$ drawing in
$\mathcal{O}(n) \times \mathcal{O}(n)$
area and that any $n$-vertex 1-plane graph admits a
RAC$_2^\text{poly}$ drawing in
$\mathcal{O}(n^3) \times \mathcal{O}(n^3)$
area.
	
We have also shown how to adjust two existing algorithms for drawing
certain 1-planar graphs such that their embeddings are preserved.  More
precisely, we have proved that any 1-plane graph admits a RAC$_1$
drawing.  This answers an open question explicitly asked by the
authors of the original algorithm~\cite{Bekos2017}.  We have also
proved that any straight-line drawable IC-plane graph admits a RAC$_0$
drawing, where the original algorithm did not necessarily preserve the
embedding~\cite{Brandenburg2016}.

The diagram in Fig.~\ref{fig:DiagramBendRACClassRelation} leaves some open questions.
Does any 1-planar graph admit a RAC$_1^\text{poly}$ drawing?
Can we draw any graph in RAC$_0$ with only right-angle crossings in polynomial area when we allow one or two bends per edge?
What is the relationship between RAC$_1$ and RAC$_2^\text{poly}$?
Can we compute RAC$_2^\text{poly}$ drawings of 1-plane graphs in~$o(n^6)$ area?
	
Figure~\ref{fig:full_drawing-after} indicates that there is
potential for postprocessing.  In particular, a compactification step
might reduce the underlying grid size considerably.  In terms of
angular resolution, our algorithm is limited by its use of the shift
algorithm.  For example, we see particularly small angles incident to the
vertices on the outer face.  Therefore, an open problem is to
find ways to improve the angular resolution without destroying any
right-angle crossings.  To this end, one may consider adding more
bends or using smooth curves for the edges.  It would also be
interesting to investigate a potential force-directed approach either
as an alternative to our algorithm or simply as a postprocessing.

%
%
\pdfbookmark[0]{References}{toc}
\section*{References}
\bibliographystyle{elsarticle-num}
\bibliography{mybibliography}

\begin{thebibliography}{10}
\expandafter\ifx\csname url\endcsname\relax
  \def\url#1{\texttt{#1}}\fi
\expandafter\ifx\csname urlprefix\endcsname\relax\def\urlprefix{URL }\fi
\expandafter\ifx\csname href\endcsname\relax
  \def\href#1#2{#2} \def\path#1{#1}\fi

\bibitem{clwz-cd1pg-GD18}
S.~Chaplick, F.~Lipp, A.~Wolff, J.~Zink, Compact drawings of 1-planar graphs
  with right-angle crossings and few bends, in: T.~Biedl, A.~Kerren (Eds.),
  Proc. 26th Int. Symp. Graph Drawing \& Network Vis. (GD'18), Vol. 11282 of
  LNCS, Springer, 2018, pp. 137--151.
\newblock \href {http://dx.doi.org/10.1007/978-3-030-04414-5_10}
  {\path{doi:10.1007/978-3-030-04414-5_10}}.

\bibitem{Purchase2000}
H.~C. Purchase, Effective information visualisation: a study of graph drawing
  aesthetics and algorithms, Interacting with Computers 13~(2) (2000) 147--162.
\newblock \href {http://dx.doi.org/10.1016/S0953-5438(00)00032-1}
  {\path{doi:10.1016/S0953-5438(00)00032-1}}.

\bibitem{Purchase2002}
H.~C. Purchase, D.~A. Carrington, J.~Allder, Empirical evaluation of
  aesthetics-based graph layout, Empirical Software Engineering 7~(3) (2002)
  233--255.
\newblock \href {http://dx.doi.org/10.1023/A:1016344215610}
  {\path{doi:10.1023/A:1016344215610}}.

\bibitem{Ware2002}
C.~Ware, H.~C. Purchase, L.~Colpoys, M.~McGill, Cognitive measurements of graph
  aesthetics, Information Visualization 1~(2) (2002) 103--110.
\newblock \href {http://dx.doi.org/10.1057/palgrave.ivs.9500013}
  {\path{doi:10.1057/palgrave.ivs.9500013}}.

\bibitem{Purchase1997}
H.~C. Purchase, Which aesthetic has the greatest effect on human
  understanding?, in: G.~D. Battista (Ed.), Proc. 5th Int. Symp. Graph Drawing
  (GD'97), Vol. 1353 of LNCS, Springer, 1997, pp. 248--261.
\newblock \href {http://dx.doi.org/10.1007/3-540-63938-1_67}
  {\path{doi:10.1007/3-540-63938-1_67}}.

\bibitem{Huang2007}
W.~Huang, Using eye tracking to investigate graph layout effects, in: S.~Hong,
  K.~Ma (Eds.), Proc. 6th Int. Asia-Pacific Symp. Visual. (APVIS'07), IEEE,
  2007, pp. 97--100.
\newblock \href {http://dx.doi.org/10.1109/APVIS.2007.329282}
  {\path{doi:10.1109/APVIS.2007.329282}}.

\bibitem{Huang2014}
W.~Huang, P.~Eades, S.~Hong, Larger crossing angles make graphs easier to read,
  J. Vis. Lang. Comput. 25~(4) (2014) 452--465.
\newblock \href {http://dx.doi.org/10.1016/j.jvlc.2014.03.001}
  {\path{doi:10.1016/j.jvlc.2014.03.001}}.

\bibitem{Huang2008}
W.~Huang, S.~Hong, P.~Eades, Effects of crossing angles, in: Proc. {IEEE}
  {VGTC} Pacific Visualization Symposium (PacificVis'08), 2008, pp. 41--46.
\newblock \href {http://dx.doi.org/10.1109/PACIFICVIS.2008.4475457}
  {\path{doi:10.1109/PACIFICVIS.2008.4475457}}.

\bibitem{Hoffmann2014}
M.~Hoffmann, M.~J. van Kreveld, V.~Kusters, G.~Rote,
  \href{http://www.cccg.ca/proceedings/2014/papers/paper05.pdf}{Quality ratios
  of measures for graph drawing styles}, in: Proc. 26th Canad. Conf. Comput.
  Geom. (CCCG'14), 2014, pp. 33--39.
\newline\urlprefix\url{http://www.cccg.ca/proceedings/2014/papers/paper05.pdf}

\bibitem{Ringel1965}
G.~Ringel, Ein {S}echsfarbenproblem auf der {K}ugel, Abh. Math. Seminar Univ.
  Hamburg 29~(1--2) (1965) 107--117.

\bibitem{KobourovLM17}
S.~G. Kobourov, G.~Liotta, F.~Montecchiani, An annotated bibliography on
  1-planarity, Comput. Sci. Rev. 25 (2017) 49--67.
\newblock \href {http://dx.doi.org/10.1016/j.cosrev.2017.06.002}
  {\path{doi:10.1016/j.cosrev.2017.06.002}}.

\bibitem{Didimo2011}
W.~Didimo, P.~Eades, G.~Liotta, Drawing graphs with right angle crossings,
  Theor. Comput. Sci. 412~(39) (2011) 5156--5166.
\newblock \href {http://dx.doi.org/10.1016/j.tcs.2011.05.025}
  {\path{doi:10.1016/j.tcs.2011.05.025}}.

\bibitem{Didimo2018}
W.~Didimo, G.~Liotta, F.~Montecchiani, A survey on graph drawing beyond
  planarity, {ACM} Comput. Surv. 52~(1) (2019) 4:1--4:37.
\newblock \href {http://dx.doi.org/10.1145/3301281}
  {\path{doi:10.1145/3301281}}.

\bibitem{s-epgg-SODA90}
W.~Schnyder, \href{http://dl.acm.org/citation.cfm?id=320176.320191}{Embedding
  planar graphs on the grid}, in: D.~S. Johnson (Ed.), Proc. 1st ACM-SIAM Symp.
  Discrete Algorithms (SODA'90), 1990, pp. 138--148.
\newline\urlprefix\url{http://dl.acm.org/citation.cfm?id=320176.320191}

\bibitem{Fraysseix1990}
H.~de~Fraysseix, J.~Pach, R.~Pollack, How to draw a planar graph on a grid,
  Combinatorica 10~(1) (1990) 41--51.
\newblock \href {http://dx.doi.org/10.1007/BF02122694}
  {\path{doi:10.1007/BF02122694}}.

\bibitem{Brandenburg2016}
F.~J. Brandenburg, W.~Didimo, W.~S. Evans, P.~Kindermann, G.~Liotta,
  F.~Montecchiani, Recognizing and drawing {IC}-planar graphs, Theor. Comput.
  Sci. 636 (2016) 1--16.
\newblock \href {http://dx.doi.org/10.1016/j.tcs.2016.04.026}
  {\path{doi:10.1016/j.tcs.2016.04.026}}.

\bibitem{Eades2013}
P.~Eades, G.~Liotta, Right angle crossing graphs and 1-planarity, Discrete
  Appl. Math. 161~(7--8) (2013) 961--969.
\newblock \href {http://dx.doi.org/10.1016/j.dam.2012.11.019}
  {\path{doi:10.1016/j.dam.2012.11.019}}.

\bibitem{Bachmaier2017}
C.~Bachmaier, F.~J. Brandenburg, K.~Hanauer, D.~Neuwirth, J.~Reislhuber,
  {NIC}-planar graphs, Discrete Appl. Math. 232 (2017) 23--40.
\newblock \href {http://dx.doi.org/10.1016/j.dam.2017.08.015}
  {\path{doi:10.1016/j.dam.2017.08.015}}.

\bibitem{Bekos2017}
M.~A. Bekos, W.~Didimo, G.~Liotta, S.~Mehrabi, F.~Montecchiani, On {RAC}
  drawings of 1-planar graphs, Theor. Comput. Sci. 689 (2017) 48--57.
\newblock \href {http://dx.doi.org/10.1016/j.tcs.2017.05.039}
  {\path{doi:10.1016/j.tcs.2017.05.039}}.

\bibitem{Harel1998}
D.~Harel, M.~Sardas, An algorithm for straight-line drawing of planar graphs,
  Algorithmica 20~(2) (1998) 119--135.
\newblock \href {http://dx.doi.org/10.1007/PL00009189}
  {\path{doi:10.1007/PL00009189}}.

\bibitem{Chrobak1995}
M.~Chrobak, T.~H. Payne, A linear-time algorithm for drawing a planar graph on
  a grid, Inf. Process. Lett. 54~(4) (1995) 241--246.
\newblock \href {http://dx.doi.org/10.1016/0020-0190(95)00020-D}
  {\path{doi:10.1016/0020-0190(95)00020-D}}.

\bibitem{Liotta2016}
G.~Liotta, F.~Montecchiani, L-visibility drawings of {IC}-planar graphs, Inf.
  Process. Lett. 116~(3) (2016) 217--222.
\newblock \href {http://dx.doi.org/10.1016/j.ipl.2015.11.011}
  {\path{doi:10.1016/j.ipl.2015.11.011}}.

\bibitem{Thomassen1988}
C.~Thomassen, Rectilinear drawings of graphs, J. Graph Theory 12~(3) (1988)
  335--341.
\newblock \href {http://dx.doi.org/10.1002/jgt.3190120306}
  {\path{doi:10.1002/jgt.3190120306}}.

\bibitem{Hopcroft1973}
J.~E. Hopcroft, R.~E. Tarjan, Algorithm 447: Efficient algorithms for graph
  manipulation, Commun. {ACM} 16~(6) (1973) 372--378.
\newblock \href {http://dx.doi.org/10.1145/362248.362272}
  {\path{doi:10.1145/362248.362272}}.

\bibitem{Zink2017}
J.~Zink,
  \href{http://www1.pub.informatik.uni-wuerzburg.de/pub/theses/2017-zink-master.pdf}{1-planar
  {RAC} drawings with bends}, Master's thesis, Institut f\"ur Informatik,
  Universit\"at W\"urzburg (2017).
\newline\urlprefix\url{http://www1.pub.informatik.uni-wuerzburg.de/pub/theses/2017-zink-master.pdf}

\bibitem{Czap2014}
J.~Czap, P.~{\v S}ugerek, \href{https://arxiv.org/abs/1404.1222}{Three classes
  of 1-planar graphs}, arXiv (2014).
\newline\urlprefix\url{https://arxiv.org/abs/1404.1222}

\bibitem{Zhang2014}
X.~Zhang, Drawing complete multipartite graphs on the plane with restrictions
  on crossings, Acta. Math. Sin. English Ser. 30~(12) (2014) 2045--2053.
\newblock \href {http://dx.doi.org/10.1007/s10114-014-3763-6}
  {\path{doi:10.1007/s10114-014-3763-6}}.

\bibitem{Zink2018}
J.~Zink,
  \href{https://github.com/j-zink-wuerzburg/embedded-graph-drawing.git}{Java
  source code: Drawing embedded graphs with right-angled crossings}, GitHub
  repository (2018).
\newline\urlprefix\url{https://github.com/j-zink-wuerzburg/embedded-graph-drawing.git}

\bibitem{wiki:Farey}
{Wikipedia contributors},
  \href{https://en.wikipedia.org/w/index.php?title=Farey_sequence&oldid=844932264}{Farey
  sequence --- {Wikipedia}{,} the free encyclopedia}, [Online; accessed
  8-June-2018] (2018).
\newline\urlprefix\url{https://en.wikipedia.org/w/index.php?title=Farey_sequence&oldid=844932264}

\bibitem{Czap2013}
J.~Czap, D.~Hud{\'{a}}k,
  \href{http://www.combinatorics.org/ojs/index.php/eljc/article/view/v20i2p54}{On
  drawings and decompositions of 1-planar graphs}, Electr. J. Comb. 20~(2)
  (2013) 54.
\newline\urlprefix\url{http://www.combinatorics.org/ojs/index.php/eljc/article/view/v20i2p54}

\bibitem{Chiba1984}
N.~Chiba, T.~Yamanouchi, T.~Nishizeki, Linear algorithms for convex drawings of
  planar graphs, in: J.~Bondy, U.~Murty (Eds.), Progress in Graph Theory,
  Academic Press, Toronto, 1984, pp. 153--173.

\end{thebibliography}

\newpage

\end{document}